\documentclass[a4paper,          % Papierformat
               12pt,             % Schriftgröße
               bibtotoc,         % Literaturverzeichnis im Inhaltsverzeichnis, ebenso Index
               idxtotoc,
               notitlepage]{scrartcl}
\usepackage[utf8]{inputenc}
\usepackage{amsmath}
\usepackage{amsthm}
\usepackage{amsfonts}
\usepackage{amssymb}

\usepackage{ifpdf}

\ifpdf
 \usepackage[pdftex]{graphicx}
 \usepackage[pdftex]{color}
\else
 \usepackage{graphicx}
 \usepackage{color}
\fi
\usepackage{verbatim}

\usepackage{ifthen}

\usepackage{makeidx}
\makeindex

\ifx\lemmasection\undefined
  \newtheorem{lemma}{Lemma}%[section]
\else
  \newtheorem{lemma}{Lemma}[section]
\fi
  \newtheorem*{lemma*}{Lemma}
  \newtheorem{prop}[lemma]{Proposition}
  \newtheorem*{prop*}{Proposition}

  \newtheorem{thm}[lemma]{Theorem}
  \newtheorem{thmA}{Theorem}
   
  \newtheorem*{thm*}{Theorem}

\theoremstyle{definition}
  \newtheorem{defi}[lemma]{Definition}
  \newtheorem*{defi*}{Definition}

\theoremstyle{remark}
  
  \newtheorem{rem}[lemma]{Remark}
  \newtheorem{ex}[lemma]{Example}
  \newtheorem*{ex*}{Example}

\newcommand{\pd}[1]{\frac{\partial}{\partial #1}}  
\newcommand{\pdd}[2]{\frac{\partial #1}{\partial #2}}  

\DeclareMathOperator{\sgn}{sgn}
\DeclareMathOperator{\supp}{supp}

\newcommand{\Ecal}{\mathcal{E}}
\newcommand{\Mcal}{\mathcal{M}}
\newcommand{\Xcal}{\mathcal{X}}
\newcommand{\Ycal}{\mathcal{Y}}
\newcommand{\Zcal}{\mathcal{Z}}

\newcommand{\Cb}{\mathbb{C}}
\newcommand{\Rb}{\mathbb{R}}
\newcommand{\Zb}{\mathbb{Z}}

\newcommand{\Phat}{\hat{P}}
\newcommand{\Qhat}{\hat{Q}}

\newcommand{\ol}{\overline}

\newcounter{Beweisschritt}
\newcommand{\Schritt}[1][\theBeweisschritt]{\setcounter{Beweisschritt}{#1}(\theBeweisschritt) \addtocounter{Beweisschritt}{1}}
\newcounter{Beweissubschritt}
\usepackage{verbatim}

\author{Johannes Rauh\footnote{Max Planck Institute for Mathematics in the Sciences, Inselstraße 22, D-04109 Leipzig, Germany}}

\title{Finding the Maximizers of the Information Divergence from an Exponential Family}

\newcommand{\KL}{information divergence}

\newcommand{\PEcal}{\ensuremath{P_{\Ecal}}}
\newcommand{\DE}{\ensuremath{D(\cdot||\Ecal)}}
\newcommand{\Hr}{\ensuremath{H_{r}}}

\newcommand{\Dbarr}{\ensuremath{\ol D_{r}}}

\newcommand{\smo}{o}

\renewcommand{\Im}{\ensuremath{\mathfrak{Im}}}
\renewcommand{\Re}{\ensuremath{\mathfrak{Re}}}

\newcommand{\keywordname}{\textit{Keywords:}}
\newcommand{\keywords}[1]{\par\addvspace\baselineskip
\noindent\keywordname\enspace\ignorespaces#1}

\begin{document}
\maketitle

\begin{abstract}
  This paper investigates maximizers of the information divergence from an exponential family $\Ecal$.  It is shown that
  the $rI$-projection of a maximizer $P$ to $\Ecal$ is a convex combination of $P$ and a probability measure $P_{-}$
  with disjoint support and the same value of the sufficient statistics $A$.  This observation can be used to transform
  the original problem of maximizing $\DE$ over the set of all probability measures into the maximization of a function
  $\Dbarr$ over a convex subset of $\ker A$.  The global maximizers of both problems correspond to each other.
  Furthermore, finding all local maximizers of $\Dbarr$ yields all local maximizers of $\DE$.
  % At present it is not known whether the reverse direction also holds.  

  This paper also proposes two algorithms to find the maximizers of $\Dbarr$ and applies them to two examples, where the
  maximizers of $\DE$ were not known before.
  % While it is not known whether the local maximizers also correspond to each other, an analysis of the first
%   order optimality conditions shows that the ``quasi-critical points'' of the two problems (defined in a convenient sense) are
%   also in bijection.

\end{abstract}

\section{Introduction}
\label{sec:intro}

Let $\Xcal$ be a finite set of cardinality $N$ and consider an \emph{exponential family} $\Ecal$ on $\Xcal$.  In this
work this will mean that there exists a real-valued $h\times N$ matrix $A$ (whose columns $A_{x}$ are indexed by
$x\in\Xcal$) and a reference measure $r$ on $\Xcal$ satisfying $r(x)>0$ for all $x\in\Xcal$ such that $\Ecal$ consists
of all probability measures on $\Xcal$ of the form
\begin{equation}
  \label{eq:expfam}
  P_{\theta}(x) = \frac{r(x)}{Z_{\theta}} \exp\left(\sum_{i=1}^{h}\theta_{i}A_{i,x}\right).
\end{equation}
In this formula $\theta\in\Rb^{h}$ is a vector of parameters and $Z_{\theta}$ ensures normalization.  The matrix $A$ is
called the \emph{sufficient statistics} of $\Ecal$.  For technical reasons it will be assumed that the row span of $A$
contains the constant vector $(1,\dots,1)$, see section \ref{sec:expfam-and-KL}.  The topological closure of $\Ecal$
will be denoted by $\ol \Ecal$.

The \emph{information divergence} (also known as the \emph{Kullback-Leibler divergence} or \emph{relative entropy}) of two probability
distributions $P$, $Q$ is defined as
\begin{equation}
  \label{eq:def:infdiv}
  D(P||Q) = \sum_{x\in\Xcal}P(x)\log\left(\frac{P(x)}{Q(x)}\right).
\end{equation}
Here we define $0\log 0 = 0\log(0/0) = 0$.  It is strictly positive unless $P=Q$, and it is infinite if the support of
$P$ is not contained in the support of $Q$.

With these definitions Nihat Ay proposed the following problem, motivated by probabilistic models for evolution and
learning in neural networks based on the infomax principle \cite{Ay02:Pragmatic_structuring}:
%Nihat~Ay proposed the following problem:
\begin{itemize}
\item Given an exponential family $\Ecal$, which probability measures $P$ maximize $D(P||\Ecal)$?
\end{itemize}
Here $D(P||\Ecal) = \inf_{Q\in\Ecal}D(P||Q)$.

Already \cite{Ay02:Pragmatic_structuring} contains a lot of properties of the maximizers, like the projection property
and support restrictions, but only in the case where the $rI$-projection $P_{\Ecal}$ of the maximizer lies in $\Ecal$.
The projection property means that the maximizer $P$ satisfies $P(x) = P_{\Ecal}(\Zcal)P_{\Ecal}(x)$ for all $x\in\Zcal
:= \supp(P)$.  In \cite{Matus07:Optimality_conditions} Franti\v{s}ek Mat\'{u}\v{s} computed the first order optimality
conditions in the general case, showing that the projection property also holds if $P_{\Ecal}\in\ol\Ecal\setminus\Ecal$.
For further results on the maximization problem see \cite{AyKnauf06:Maximizing_Multiinformation,
  Matus09:Divergence_from_factorizable_distributions, MatusAy03:Maximization_of_the_Info_Div_from_Exp_Fam}.

In this work it is shown that the original maximization problem can be solved by studying the related problem:
\begin{itemize}
\item Maximize the function $\Dbarr(u) = \sum_{x\in X}u(x)\log\frac{|u(x)|}{r(x)}$ for all $u\in\ker A$ such that
  $||u||_{1}\le 2$ and $\sum_{x}u_{x}=0$.
\end{itemize}
Here, $||u||_{1}$ is the $\ell_{1}$-norm of $u$.  Theorem \ref{thm:Dualmaxi} will show that there is a bijection between
the global maximizers of these two maximization problems.  Furthermore, knowing all local maximizers of $\Dbarr$ yields
all local maximizers of $\DE$.  This relation is a consequence of the projection property mentioned above.

In Section \ref{sec:expfam-and-KL} some known properties of exponential families and the \KL{} are collected, including
Matú\v{s}'s result on the first order optimality conditions of maximizers of $\DE$.  In Section \ref{sec:projpoints} the
projection property is analyzed.  It is easy to see that probability measures that satisfy the projection property and
that do not belong to $\Ecal$ come in pairs $(P_{+},P_{-})$ such that $P_{+}-P_{-}\in\ker A\setminus\{0\}$.  This
pairing is used in Section \ref{sec:KernelDecomp} to replace the original problem by the maximization of the function
$\Dbarr$.
%
%another problem: The new problem is to maximize a function $\Dbarr$ which is defined on $\ker A\setminus\{0\}$, and
%there is a one-to-one-correspondence of the global maximizers of both problems.  
% One advantage of the reformulation is that there is a formula for $\Dbarr$, while $\DE$ is defined by a minimization
% problem.
Theorem \ref{thm:Dualmaxi} in this section investigates the relation between the maximizers of both problems. %
In Section \ref{sec:firstorder} the first order conditions of $\Dbarr$ are computed.
% It is shown in Section \ref{sec:critpts} that the \emph{quasi-critical points} (defined in a suitable way) of both
% maximization problems are also in a one-to-one-correspondence.
Section \ref{sec:codim1} discusses the case where $\dim\ker A =1$, demonstrating how the reformulation leads to a quick
solution of the original problem.
Section~\ref{sec:algo} gives some ideas how to solve the critical equations from Section \ref{sec:firstorder}.
Section~\ref{sec:projpts} presents an alternative method of computing the local maximizers of $\DE$, which uses the
projection property more directly.  Sections~\ref{sec:algo} and~\ref{sec:projpts} contain two examples which demonstrate
how the theory of this paper can be put to practical use.

\section{Exponential families and the \KL}
\label{sec:expfam-and-KL}

The definition of an exponential family, as it will be used in this work, was already stated in the introduction.  It is
important to note that the correspondence between exponential families $\Ecal$ on one side and sufficient statistics $A$
and reference measure $r$ on the other side is not unique.  One reason for this lies in the normalization of probability
measures: We can always add a constant row to the matrix $A$ without changing $\Ecal$ (as a set).  For this reason in
the following it will be assumed that $A$ contains the constant row $(1,\dots,1)$ in its row space.  This implies that
every $u\in\ker A$ satisfies $\sum_{x\in\Xcal}u(x)=0$.

In order to characterize the remaining ambiguity in the parametrization $(r,A)\mapsto\Ecal$, denote by $\Ecal_{r,A}$ the
exponential family associated to a given matrix $A$ and a given reference measure.  Then $\Ecal_{r,A}=\Ecal_{r',A'}$ as
sets if and only if the following two conditions are satisfied:
\begin{itemize}
\item $r\in\Ecal_{r',A'}$.
\item The row span of $A$ equals the row span of $A'$.
\end{itemize}

The introduction also featured the definition of the \KL.  In the following we will also use formula
\eqref{eq:def:infdiv} for positive measures $Q$ which are not necessarily normalized.  In this case
\begin{equation}
  \label{eq:infdiv-lambda}
  D(P||\lambda Q) = \sum_{x}P(x)\log\frac{P(x)}{\lambda Q(x)} = D(P||Q) - \log\lambda\quad\text{ for all }\lambda > 0,
\end{equation}
where $\sum_{x}P(x)=1$ was used.

The following theorem sums up the main facts about exponential families:
\begin{thmA}
  \label{thmA:rI-projection}
  Let $P$ be a probability measure on $\Xcal$.  Then there exists a unique probability measure $P_{\Ecal}$ in $\ol
  \Ecal$ such that $AP = AP_{\Ecal}$.  Furthermore, $P_{\Ecal}$ has the following properties:
  \begin{enumerate}
  \item For all $Q\in\Ecal$
    \begin{equation}
      \label{eq:Pythident}
      D(P||Q) = D(P||P_{\Ecal}) + D(P_{\Ecal}||Q).
    \end{equation}
  \item $\PEcal$ satisfies
    \begin{equation}
      \label{eq:D=H-H}
      D(P||\Ecal) = H_{r}(\PEcal) - H_{r}(P)
    \end{equation}
  % \item $P_{\Ecal}$ is the MLE\dots
  \item $P_{\Ecal}$ maximizes the concave function
    \begin{equation}
      \Hr(Q) := - \sum_{x} Q(x)\log \frac{Q(x)}{r(x)}
    \end{equation}
    subject to the condition $AQ = AP$.
  \end{enumerate}
\end{thmA}
\begin{proof}[Sketch of proof]
  Corollary 3.1 of \cite{CsiszarShields04:Information_Theory_and_Statistics} proves existence and uniqueness of
  $P_{\Ecal}$ and the ``Pythagorean identity''~\eqref{eq:Pythident}
  % \begin{equation}
  %   D(P||Q) = D(P||P_{\Ecal}) + D(P_{\Ecal}||Q)
  % \end{equation}
  for all probability measures $P$ and all probability measures $Q\in\ol \Ecal$.  It follows from~\eqref{eq:infdiv-lambda}
  that
  \begin{equation}
    D(P||r) = D(P||P_{\Ecal}) + D(P_{\Ecal}||r),
  \end{equation}
  so statements~{2.} and~{3.} follow from $H_{r}(Q) = -D(Q||r)$.
\end{proof}
$P_{\Ecal}$ is called the \emph{$rI$-projection} of $P$ to $\Ecal$, or simply the \emph{projection} of $P$ to $\Ecal$.

Note that the function $H_{r}$ introduced in the theorem satisfies $H_{r}(P) = -D(P||r)$.  It can thus be interpreted as
a negative \emph{relative entropy}.
% \footnote{Actually, it would be nice to give $H_{r}$ itself the name ``relative
%   entropy'', since it is much more directly related to the entropy than the information divergence.  Unfortunately this
%   goes against the common information theoretical shibboleth.}.
In this work $H_{r}$ is prefered to its negative counterpart in order to keep the connection to the entropy $H$ visible
in the important case that $r(x)=1$ for all $x\in\Xcal$.

% Equation \eqref{eq:D=H-H} gives a first idea how the maximizer should look like: It should have a low entropy, which
% hints at a small support, while its projection should have a large entropy.

The map associated to the matrix $A$ is called the \emph{moment map}.  It maps the set of all probability measures on
$\Xcal$ onto the polytope $\Mcal$ which is the convex hull of the columns of $A$.  This polytope is called the
\emph{convex support} of $\Ecal$.  In the special case that $\Ecal$ is a hierarchical model (see
\cite{Lauritzen96:Graphical_Models}), $\Mcal$ is called the \emph{marginal polytope} of $\Ecal$.

Note that we can associate a point $A_{x}\in\Mcal$ with each state $x\in\Xcal$.  Among these points are the vertices of
$\Mcal$, but not every point $A_{x}$ needs to be a vertex of $\Mcal$.

% \begin{thmA}
%   \label{thmA:momentmap}
%   The moment map induces a homeomorphism $\ol \Ecal\to\Mcal$.  It satisfies the following combinatorial property:
%   \begin{itemize}
%   \item Let $P\in\Ecal$.  Then the image $AP$ lies in the relative interior of a unique face of $\Mcal$.  The vertices
%     of this face are given by $\supp(P)\subseteq\Xcal$.
%   \end{itemize}
% \end{thmA}
% \begin{proof}
%   See ***.
% \end{proof}

\begin{thmA}
  \label{thmA:Matus}
  Let $P_{+}$ be a (local) maximizer of $D(\cdot||\Ecal)$ with support $\Zcal=\supp(P_{+})$ and $P_{\Ecal}$ its
  $rI$-projection to $\Ecal$. Then the following holds:
  \begin{enumerate}
  \item $P_{+}$ satisfies the \emph{projection property}, i.e., up to normalization $P_{+}$ equals the restriction of
    $P_{\Ecal}$ to $\Zcal$:
    \begin{equation}
      \label{eq:defcritpt}
      P_{+}(x) = 
      \begin{cases}
        \frac{\PEcal(x)}{\PEcal(\Zcal)}, & \text{ if } x\in \Zcal,\\
        0, & \text{ else.}
      \end{cases}
    \end{equation}
  \item Suppose $\Ycal := \supp(P_{\Ecal})\neq\Xcal$.  Then the moment map maps $\Ycal$ and $\Xcal\setminus\Ycal$ into
    parallel hyperplanes.
  \item The cardinality of $\Zcal$ is bounded by $\dim\Ecal + 1$.
  \end{enumerate}
\end{thmA}
\begin{proof}
  Statements {1.} and {3.} were already known to Ay\cite{Ay02:Pragmatic_structuring} in the special case where
  $\Ycal=\Xcal$.  The general form of statement {3.} is Proposition 3.2 of
  \cite{MatusAy03:Maximization_of_the_Info_Div_from_Exp_Fam}.  Statement {2.} and the generalization of statement {1.}
  are due to Matúš\cite[Theorem 5.1]{Matus07:Optimality_conditions}.
\end{proof}
The paper \cite{Matus07:Optimality_conditions} contains further conditions on the maximizer.  However, these
will not be studied in this work.
\begin{defi}
  \label{def:critpt}
  Any probability measure $P$ that satisfies \eqref{eq:defcritpt} will be called a \emph{projection point}.  If $P$
  satisfies conditions {1.} and {2.} of Theorem \ref{thmA:Matus}, then $P$ will be called a \emph{quasi-critical point}
  of $D(\cdot||\Ecal)$, or a \emph{$D$-quasi-critical point}\footnote{In convex analysis, a point satisfying all
    first-order conditions (which in general comprise both equations and inequalities) of a convex function is called a
    \emph{critical point}.  In analogy to this, the term ``quasi-critical'' point is chosen in this work for a point
    which satisfies only the \emph{equations} derived from the first order conditions of an arbitrary function.}.
\end{defi}

\section{Projection points}
\label{sec:projpoints}

In this section assume that $A$ does not have full rank.  Otherwise the function $\DE$ is trivial.

Let $P_{+}$ be a projection point, and let $\PEcal$ be its projection to $\Ecal$.  Denote $\Zcal=\supp(P_{+})$ and
$\Ycal=\supp(P_{\Ecal})$.  Every measure $P_{\lambda} := \lambda P_{+} + (1 - \lambda) \PEcal$ on the line through
$P_{+}$ and $\PEcal$ is normalized and has the same sufficient statistics as $P_{+}$ and $\PEcal$. Fix $\lambda_{-} = -
\frac{\PEcal(\Zcal)}{1 - \PEcal(\Zcal)}$. Then
\begin{equation}
  P_{\lambda_{-}}(x) =
  \begin{cases}
    - \frac{\PEcal(\Zcal)}{1 - \PEcal(\Zcal)} \frac{\PEcal(x)}{\PEcal(\Zcal)} + \frac{1}{1 - \PEcal(\Zcal)} \PEcal(x) = 0 & \text{ if }x \in \Zcal, \\
    (1 - \lambda_{-}) \PEcal(x) = \frac{1 + \PEcal(\Zcal)}{1 - \PEcal(\Zcal)} \PEcal(x) \ge 0                   & \text{ else.}
  \end{cases}
\end{equation}
Thus $P_{-}:= P_{\lambda_{-}}$ is a probability measure with support equal to $\Ycal \setminus \Zcal$, and $u :=
P_{+} - P_{-}$ lies in the kernel of $A$.
%
%In order to keep the notation symmetric we will write $P_{+}:= P$ in the following. 
%We have the following relations between $P_{+}$, $P_{-}$ and $\PEcal$:
Furthermore, $P_{-}$ is a second projection point with the same projection $\PEcal$ to $\Ecal$ as
$P_{+}$.

The projection $\PEcal$ can be written as a convex combination of $P_{+}$ and $P_{-}$, i.e., $\PEcal = \mu P_{+} + (1-\mu) P_{-}$,
where $\mu = \frac{\lambda_{-}}{\lambda_{-}-1} \in (0, 1)$. Since the supports of $P_{+}$ and $P_{-}$ are disjoint we
have $\mu = \PEcal(\Zcal)$ and $(1-\mu) = \PEcal(\Xcal\setminus \Zcal)$. In other words,
\begin{equation}
  \PEcal(x) =
  \begin{cases}
    \mu P_{+}(x), & x\in \Zcal, \\
    (1 - \mu) P_{-}(x), & x\notin \Zcal.
  \end{cases}
\end{equation}

There are a lot of relations between $P_{+}$, $P_{-}$ and $\PEcal$.  They will be collected in the following Lemma in a
slightly more general form.
\begin{lemma}
  \label{lem:formulae}
  Let $P_{+}$ and $P_{-}$ be two probability measures with disjoint supports such that $AP_{+}=AP_{-}$.  Let $\Phat$ be
  the unique probability measure in the convex hull of $P_{+}$ and $P_{-}$ that maximizes the function
  \begin{equation}
    \Hr(Q) = - \sum_{x} Q(x)\log \frac{Q(x)}{r(x)}.
  \end{equation}
  Define $\mu=\Phat(\Zcal)$, where $\Zcal=\supp(P_{+})$.
  Then the following equations hold:
  \begin{subequations}
    \begin{equation}
      \label{eq:HPhatHp}
      \exp(\Hr(\Phat)) = \exp(\Hr(P_{+})) + \exp(\Hr(P_{-})),
    \end{equation}
    \begin{equation}
      \label{eq:mufrac}
      \frac{\mu}{1-\mu} = \exp\left(\Hr(P_{+}) - \Hr(P_{-})\right),
    \end{equation}
    \begin{equation}
      \label{eq:KL}
      D(P_{+}||\Phat) = \Hr(\PEcal) - \Hr(P_{+}) = \log(1 + \exp(\Hr(P_{-}) - \Hr(P_{+}))).
    \end{equation}
  \end{subequations}
\end{lemma}
\begin{proof}
  The first observation is
% \begin{equation}
%   \label{eq:HPhat_mu}
%   H(\PEcal) = \mu H(P_{+}) + (1-\mu) H(P_{-}) + h(\mu,1-\mu),
% \end{equation}
% where $h(\mu,1-\mu) = -\mu \log (\mu) - (1-\mu) \log (1-\mu)$. The same equation holds for the \emph{relative entropy}
% $\Hr(\PEcal) := - \sum_{x} \PEcal(x) \log \frac{\PEcal(x)}{r(x)} = - D(\PEcal||r)$:
\begin{equation}
  \label{eq:HrPhat_mu}
  \Hr(\PEcal) = \mu \Hr(P_{+}) + (1-\mu) \Hr(P_{-}) + h(\mu,1-\mu),
\end{equation}
where $h(\mu,1-\mu) = -\mu \log (\mu) - (1-\mu) \log (1-\mu)$.

Since $\PEcal$ maximizies $\Hr$ among all probability measures with the same sufficient statistics as $P_{+}$ and
$P_{-}$, it follows that
\begin{multline*}
  \left.\pdd{\left(\mu' \Hr(P_{+}) + (1-\mu') \Hr(P_{-}) + h(\mu',1-\mu')\right)}{\mu'}\right|_{\mu'=\mu} 
  =  \\  =
  \Hr(P_{+}) - \Hr(P_{-}) + \log(1-\mu) - \log(\mu)
\end{multline*}
must vanish, which rewrites to
\begin{equation}
%  \label{eq:mufrac}
  \frac{\mu}{1-\mu} = \exp\left(\Hr(P_{+}) - \Hr(P_{-})\right),
\end{equation}
or
\begin{equation}
  \label{eq:mu_Hs}
  \mu = \frac{\exp(\Hr(P_{+}))}{\exp(\Hr(P_{+})) + \exp(\Hr(P_{-}))} = \frac{1}{1 + \exp(\Hr(P_{-}) - \Hr(P_{+}))}.
\end{equation}
This implies
\begin{align*}
  h(\mu,1-\mu) & = - \mu \Hr(P_{+}) + \mu \log\left(\exp(\Hr(P_{+})) + \exp(\Hr(P_{-}))\right) \\ &\qquad
                   - (1-\mu) \Hr(P_{-}) + (1-\mu) \log\left(\exp(\Hr(P_{+})) + \exp(\Hr(P_{-}))\right) \\
               & = - \mu \Hr(P_{+}) - (1-\mu) \Hr(P_{-}) + \log\left(\exp(\Hr(P_{+})) + \exp(\Hr(P_{-}))\right).
\end{align*}
Comparison with equation \eqref{eq:HrPhat_mu} yields
\begin{equation}
%  \label{eq:HPhatHp}
  \exp(\Hr(\PEcal)) = \exp(\Hr(P_{+})) + \exp(\Hr(P_{-})),
\end{equation}
which in turn simplifies \eqref{eq:mu_Hs} to
\begin{equation}
  \mu = \exp(\Hr(P_{+}) - \Hr(\PEcal)).
\end{equation}
The Kullback-Leibler divergence equals
%$D(P_{+}||\Ecal)$ equals
\begin{subequations}
%  \label{eqs:KL}
  \begin{align}
%    \label{eq:KLmu}
    D(P_{+}||\PEcal) = \sum_{x\in \Zcal} P_{+}(x) \log \frac{1}{\PEcal(\Zcal)} &= - \log(\mu) \\
%    \label{eq:KLPhat}
                                                                             &= \Hr(\PEcal) - \Hr(P_{+}) \\
%    \label{eq:KLH}
                                                                             &= \log(1 + \exp(\Hr(P_{-}) - \Hr(P_{+}))).
  \end{align}
\end{subequations}
\end{proof}
%This means that, in order to maximize $D(P_{+}||\Ecal)$ we have to 
%Equality \eqref{eq:KLPhat} is, of course, Basti's equality.

As an easy consequence
\begin{equation}
  \exp(-D(P_{+}||\Ecal)) + \exp(-D(P_{-}||\Ecal)) = 1,
\end{equation}
from which we see that in general $P_{+}$ and $P_{-}$ will not be both maximizers of $D(\cdot||\Ecal)$.  Furthermore it
follows that $D(P||\Ecal) \ge \log(2)$ for any global maximizer $P$ (assuming that $A$ does not have full rank).

\section{Decomposition of Kernel Elements}
\label{sec:KernelDecomp}

Now suppose that $u$ is an arbitrary nonzero element from the kernel of $A$. Then $u = u_{+} - u_{-}$, where $u_{+}$ and
$u_{-}$ are positive vectors of disjoint support.  Since $A$ contains the constant vector $(1,\dots,1)$ in its rowspan,
it follows that the $\ell_{1}$-norms of $u_{+}$ and $u_{-}$ are equal.  Thus $u = d_{u}(P_{+} - P_{-})$, where $d_{u} =
||u_{+}||_{1} = ||u_{-}||_{1} = \frac12 ||u||_{1} > 0$ is called the \emph{degree} of $u$ and $P_{+}$ and $P_{-}$ are
two probability measures with disjoint supports.  Since $P_{+}$ and $P_{-}$ have the same image under~$A$, they have the
same projection to~$\Ecal$, which will be denoted by~$P_{\Ecal}$.

Let $\Phat$ be the convex combination of $P_{+}$ and $P_{-}$ that maximizes $H_{r}$.  Note that in general $\Phat\neq
P_{\Ecal}$.  Still Lemma \ref{lem:formulae} applies.
% However, since we never used the fact that $\PEcal$ is part of the exponential family (or its
% closure), all the formulas in the previous section remain valid, if we replace $\PEcal$ by $\Phat$, with the only
% exception that in general $D(P_{+}||\Ecal) = D(P_{+}||P_{\Ecal})$ will be different from $D(P_{+}||\Phat)$ in equation
% \eqref{eq:KLmu}. In fact, we have
Furthermore
%\begin{lemma}
\begin{equation}
  \label{eq:DEPhat}
  D(P_{+}||\Ecal) = \Hr(P_{\Ecal}) - \Hr(P_{+}) 
  \ge D(P_{+}||\Phat) = \Hr(\Phat) - \Hr(P_{+}),
\end{equation}
% \end{lemma}
since $P_{\Ecal}$ maximizes $\Hr$ when the image under $A$ is constrained (see Theorem~\ref{thmA:rI-projection}).

% which can be seen as follows:
% %\begin{proof}
%   Consider the exponential family $\Ecal_{M}$ obtained in the following way: Add further rows to $A$ to obtain a
%   matrix $A_{M}$ such that $\ker A_{M}$ is spanned from $M$.  Then $\Phat = P_{\Ecal_{M}}$, and since
%   $\Ecal\subseteq\Ecal_{M}$ we have
%   \begin{equation}
%     D(P_{+}||\Phat) = D(P_{+}||\Ecal_{M}) \le D(P_{+}||\Ecal) = D(P_{+}||\PEcal).
%   \end{equation}
% %\end{proof}

These facts can be used to relate two different optimization problems. The first one is the maximization of the \KL{}
from $\Ecal$. The second one is the maximization of the function
\begin{equation}
  \label{eq:H1}
  \Dbarr: \ker A %\setminus \{0\}
  \to \Rb, u \mapsto \sum_{x} u(x)\log\frac{|u(x)|}{r(x)}
  % \Hr(P_{-}) - \Hr(P_{+}),
\end{equation}
subject to the constraint $d_{u} = \frac{1}{2}||u||_{1} = 1$.  From what has been said above, if $d_{u}=1$ then
$u=Q_{+}-Q_{-}$ for two probability measures $Q_{+},Q_{-}$ with disjoint support, and in this case
\begin{equation}
  \Dbarr(u) = \Hr(Q_{-}) - \Hr(Q_{-}).
\end{equation}
% where $P_{\pm}$ are the two probability measures with disjoint support induced by $u$.
Since $\Dbarr$ is a continuous function from the compact $\ell_{1}$-sphere of radius $2$ in $\ker A$, a maximum is
guaranteed to exist.
% These facts lead to the following theorem:

%We have the following theorem:
\begin{thm}
  \label{thm:Dualmaxi}
  Let $\Ecal$ be an exponential family with sufficient statistics $A$.
  \begin{enumerate}
%  \item Let $M \in \ker A \setminus\{0\}$ be a local maximizer of $\Hrone$. If $M$ has full support, then $P_{+}$ is a local maximizer of $D(\cdot||\Ecal)$.
  \item %If $M$ is a local minimizer of $\Hrone$, then $P_{+}$ is a local maximizer of $D(\cdot||\Ecal)$.
    If $u = Q_{+}-Q_{-} \in \ker A \setminus\{0\}$ is a global maximizer of $\Dbarr$ subject to $d_{u} =
    \frac{1}{2}||u||_{1} = 1$, then the positive part $Q_{+}$ of $u$ globally maximizes $D(\cdot||\Ecal)$.
  \item Let $P_{+}$ be a local maximizer of the information divergence.  There exists a unique probability measure
    $P_{-}$ with support disjoint from $P_{+}$ such that $P_{+} - P_{-} \in \ker A$ is a local maximizer of $\Dbarr$.
    If $P_{+}$ is a global maximizer, then $P_{+} - P_{-}$ is a global maximizer.
%  \item Let $M = P_{+} - P_{-}$ be an element of the kernel of $A$, where $P_{+}$ and $P_{-}$ are probability measures
%    with disjoint support, such that $H(P_{-}) - H(P_{+})$ is maximal. Then $P_{+}$ is a global maximizer of the
%    information divergence from $\Ecal$.
%  \item Let $P$ be a global maximizer of the information divergence. Then there exists a probability measure $P_{-}$
%    with support disjoint from $P$ such that $P - P_{-}$ lies in the kernel of $A$, and $H(P_{-}) - H(P)$ is
%    maximal.
%  \item The correspondence between maximizers of $M\in\ker A \mapsto H(P_{-}) - H(P_{+})$ and maximizers of
%    $D(\cdot||\Ecal)$ also works for local maximizers such that $\Phat$ (defined either by $M$ or equivalently by $P$)
%    has full support.
  \end{enumerate}
\end{thm}
\begin{proof}
  \Schritt[1]
  Consider global maximizers first:

  Choose probability measures $Q_{+}$ and $Q_{-}$ of disjoint support such that $u = Q_{+} - Q_{-}$ maximizes
  $\Dbarr$. %satisfies the premises of the first part of the theorem.
  Denote by $\Qhat$ the probability measure from the convex hull of $Q_{+}$ and $Q_{-}$ that maximizes $\Hr$. %
  In addition, let $P_{+}$ be a global maximizer of $D(\cdot||\Ecal)$. Construct $P_{-}$ as in section
  \ref{sec:projpoints}.
  %in the first part of this section.
 %Furthermore $\Qhat'$ shall be the projection of
 % $Q_{+}$, $Q_{-}$ and $\Qhat$ to $\Ecal$. Since $H(\Qhat')$ is maximal among all probability measures that share the
 % sufficient statistics with $Q_{+}$, $Q_{-}$ and $\Qhat$, we have $H(\Qhat') \ge H(\Qhat)$.
%
  From \eqref{eq:KL} and \eqref{eq:DEPhat} it follows that
  \begin{align}
    \log(1 + \exp(\Hr(P_{-}) - \Hr(P_{+}))) &= D(P_{+}||\Ecal) \ge D(Q_{+}||\Ecal) \nonumber\\
  %                  &= D(Q_{+}||\Qhat') \nonumber\\
  %                  &= H(\Qhat') - H(Q_{+}) \nonumber\\
                    &\ge D(Q_{+}||\Qhat) = \Hr(\Qhat) - \Hr(Q_{+}) \nonumber\\
                    &= \log(1 + \exp(\Hr(Q_{-}) - \Hr(Q_{+}))).
    \label{eq:proofofthm}
  \end{align}
  The maximality property of $Q_{+}-Q_{-}$ implies that all terms of \eqref{eq:proofofthm} are equal. This proves the
  global part of the theorem.

  % \Schritt Suppose that $P=P_{+}-P_{-}$ is a local minimizer of $\Hrone$.  Then $P_{-}$ maximizes the entropy $\Hr$
  % subject to the conditions $P_{-}\in\P_{+} + \ker A$ and $\supp P_{-}\cap\supp P_{+}=\emptyset$.
  %
  \Schritt
%  For the local statement one can argue as follows:
%  Use notation as in the previous step. 
  Now suppose that $P_{+}$ is a local maximizer of the information divergence and define $P_{-}$ as above. 
%Then $P_{-}$ is uniquely defined. 
  Choose a neighbourhood $U$ of $P_{-}$ such that $D(P'||\Ecal) \le D(P_{+}||\Ecal)$ for all $P'\in U$. Since the map
  $u\mapsto (u_{+},u_{-})$ %(where $u = d_{u}(Q_{+}-Q_{-})$)
  is continuous, there is a neighbourhood $U'$ of $P_{+}-P_{-}$ such that $Q'_{+} - Q'_{-} \in U' \Longrightarrow
  Q'_{+}\in U$ for all probability measure $Q'_{+},Q'_{-}$.  It follows that
  \begin{multline}
    \log(1 + \exp(\Hr(P_{-}) - \Hr(P_{+}))) = D(P_{+}||\Ecal) \ge D(Q'_{+}||\Ecal) \nonumber\\
                    \ge \Hr(\Qhat') - \Hr(Q_{+}') = \log(1 + \exp(\Hr(Q_{-}') - \Hr(Q_{+}'))).
%    \label{eq:proofofthm}
  \end{multline}
  for all $Q'_{+} - Q'_{-}$ from the neighbourhood $U'$ of $P_{+} - P_{-}$. Thus $P_{+} - P_{-}$ is a local maximizer.

  $P_{-}$ is unique since it is characterized as the unique maximizer of the concave function $\Hr$ under the linear
  constraints $P_{+} - P_{-}\in\ker A$ and $\supp(P_{+})\cap \supp(P_{-}) = \emptyset$.
  %
  \begin{comment}
    Now let $M = Q_{+} - Q_{-}$ be a local maximizer of $M\in\ker A \mapsto \Hr(Q_{-}) - H(Q_{+})$.

    Suppose that $\Qhat$ is not the projection of $Q_{+}$ to $\Ecal$.

    Suppose that $P$ is not a local maximizer of $D(\cdot||Q)$. This means that any neighbourhood $U$ of $P$ contains a
    probability measure

    If $\supp (P - P_{-}) = \supp (\Phat)$ is full, then there is a neighbourhood $U$ of $P - P_{-}$ in $\ker A$ with
    the following property: If $Q' = Q_{+}' + Q_{-}' \in U$, then $\supp(Q_{+}') = \supp (P)$ and $\supp (Q_{-}') =
    \supp(P_{-})$. Furthermore $U$ can be chosen such that $D(Q'_{+}||\Ecal) \le D(P||\Ecal)$ for all $Q'\in U$. This
    means that the chain of inequalities \eqref{eq:proofofthm} can be applied as before to show that $P - P_{-}$ locally
    maximizes $H(\cdot_{-}) - H(\cdot_{+})$.
  \end{comment}
\end{proof}

\begin{rem}
  \label{sec:reformulatingDbarr}
  There are several possibilities to reformulate the problem of maximizing $\Dbarr$.
  % Let $u = d_{u}(P_{-} - P_{+})\in\ker A\setminus\{0\}$.  Then
  To see this, note that $\Dbarr$
  % \begin{equation}
  %   \Dbarr(u) = H_{r}(P_{-}) - H_{r}(P_{+}) = \sum_{x} \frac{M(x)}{d_{M}} \log \frac{|M(x)|}{d_{M}r(x)}
  % \end{equation}
  % It follows that $\Dbar$
  is homogeneous of degree one, since
  \begin{equation}
    \Dbarr(\alpha u) = \alpha \sum_{x}u(x)\log \frac{|u(x)|}{r(x)} + \alpha \left(\sum_{x} u(x)\right) \log|\alpha| = \alpha \Dbarr(u)
  \end{equation}
  for all $u\in\ker A$ and $\alpha\in\Rb$.  This means that, when maximizing $\Dbarr$, the constraint $d_{u}=1$ is
  equivalent to $d_{u}\le 1$.  Under the inequality constraint the maximization is over a polytope, while under the
  equality constraint the maximization is over the boundary of the same polytope.

  A third alternative is the maximization of the function
  \begin{equation}
    \label{eq:Dbarr1}
    \Dbarr^{1}:\; \ker A\setminus\{0\} \to \Rb,\quad u \mapsto \frac{1}{d_{u}}\Dbarr(u).
  \end{equation}
  The solutions of this last problem need to be normalized in order to compare this maximization problem with the
  formulations.
  % In other words,
  % \begin{equation}
  %   \Dbar(M) = \frac{1}{d_{M}}\sum_{x} M(x) \log \frac{|M(x)|}{r(x)}.
  % \end{equation}
\end{rem}
\begin{rem}
  It is an open question when the projection $\PEcal$ of a maximizer $P_{+}$ lies in the interior of the probability
  simplex.  More generally one could ask for the support of $\PEcal$.  Since $\supp(\PEcal) = \supp(P_{+}-P_{-})$ this
  question can also be studied with the help of the theorem.

  In many cases the support of $\PEcal$ will be all of $\Xcal$.  However, the construction of Example \ref{ex:1} shows
  that $\PEcal$ can have any support (of cardinality at least two).  See also \cite{Matus07:Optimality_conditions}.
\end{rem}

\section{First order conditions}
\label{sec:firstorder}

Theorem \ref{thm:Dualmaxi} implies that all maximizers of $\DE$ are known once all maximizers of $\Dbarr$ are found.
The latter can be computed by solving the first order conditions.
% calculate the first order variation of $\Dbarr$ in this section.
To simplify the notation define
\begin{equation}
  u(B) := \sum_{x\in B} u(x)
\end{equation}
if $u\in\Rb^{\Xcal}$ is any vector and $B\subseteq\Xcal$.

\begin{prop}
  \label{prop:critpts}
  Let $u\in\ker A$ be a local maximizer of $\Dbarr$ subject to $d_{u}=\frac{1}{2}||u||_{1}$. %mit $d(P)=1$.
  The following statements hold:
  \begin{enumerate}
  \item 
    \label{it:Var0}
%    \begin{equation}
    \label{eq:Var0}
    $v(u=0) := \sum_{x:u(x)=0}v(x) = 0$
%    \end{equation}
    for all $v\in\ker A$.
  \item 
    \label{it:Var1ineq}
%    and
    $u$ satisfies
    \begin{equation}
      \label{eq:Var1ineq}
      \sum_{x:u\neq 0} v(x) \log \frac{|u(x)|}{r(x)} + \sum_{x:u=0} v(x) \log \frac{|v(x)|}{r(x)} \ge d'_{u}(v) \Dbarr(u)
    \end{equation}
%    hold
    for all $v\in\ker A$, where $d_{u}'(v) := v(u>0) + v_{+}(u=0)$.
  \item If $v\in\ker A$ satisfies $\supp (v) \subseteq \supp (u)$, then
    \label{it:Var1}
    \begin{equation}
      \label{eq:Var1}
      \sum_{x:u\neq 0} v(x) \log \frac{|u(x)|}{r(x)} = d'_{u}(v) \Dbarr(u).
    \end{equation}
  \end{enumerate}
\end{prop}
\begin{proof}
  % Now let $u = u_{+} - u_{-} = d_{u} (P_{+} - P_{-}) \in \ker A$. 
  First note that the {degree} $d_{v} = \sum_{x} v_{+}(x) = \sum_{x} v_{-}(x) = \frac{1}{2} ||v||_{1}$
  % of $u$ satisfies
  % \begin{equation}
  %   d_{u} = \sum_{x} u_{+}(x) = \sum_{x} u_{-}(x) = \frac{1}{2} ||u||_{1}.
  % \end{equation}
  % It 
  is piecewise linear in the following sense:
  \begin{itemize}
  \item Let $u,v\in\ker A$. Then there exists $\lambda_{1}>0$ such that
    \begin{equation}
      \label{eq:loclind}
      d_{u + \lambda v} = d_{u} + \lambda d'_{u}(v) \text{ for all }0\le\lambda\le\lambda_{1},
    \end{equation}
    where $d'_{u}(v) = \sum_{x:u>0} v(x) + \sum_{x:u=0} v_{+}(x) = v(u>0) + v_{+}(u=0) \in\Rb$ depends only on $u$ and $v$
    (but not on $\lambda$).
  \end{itemize}

  % For any $v\in\ker A$ the map $\epsilon\to\frac{1}{d_{u+\epsilon v}}(u + \epsilon v)$ is a path 

  % We will call
  % \begin{equation}
  %   \Hr(M) := - \sum_{x} M(x) \log \frac{|M(x)|}{r(x)}
  % \end{equation}
  % the \emph{(relative) entropy} of $M$.

  % For a moment, define
  % \begin{equation}
  %   \Hr(v) := - \sum_{x} v(x) \log \frac{|v(x)|}{r(x)}.
  % \end{equation}
  % for all $v\in\ker A$.  
  Fix $u, v\in\ker A$.  If $\epsilon>0$ is small enough then
  \begin{align*}
    \Dbarr(u + \epsilon v) &= \sum_{x} u(x) \log \frac{|u(x)|}{r(x)} + \sum_{x:u\neq 0} u(x) \log \left(1 + \epsilon
      \frac{v(x)}{u(x)} \right)
    \\
    & \qquad + \epsilon \sum_{x} v(x) \log \frac{|u(x) + \epsilon v(x)|}{r(x)}
    \\
    % & = \Dbarr(u) + \epsilon \sum_{x:u\neq 0} v(x)
    % \\
    % & \qquad + \epsilon \sum_{x:u\neq 0} v(x) \log \frac{|u(x)|}{r(x)} + \epsilon \sum_{x:u=0} v(x) \log \frac{|\epsilon
    %   v(x)|}{r(x)} + \smo(\epsilon)
    % \\
    & = \Dbarr(u) + \epsilon \left(\sum_{x:u\neq 0} v(x) \log \frac{|u(x)|}{r(x)} + \sum_{x:u=0} v(x) \log \frac{|v(x)|}{r(x)}\right)\\
    & \qquad + \epsilon \log |\epsilon| v(u = 0) %\sum_{x:u=0} v(x)
    + \epsilon v(u\neq 0) + \smo(\epsilon),
  \end{align*}
  where $\log(1+\epsilon x) = 1 + \epsilon x + \smo(\epsilon)$ was used.
  % In the following we abbreviate $h(\epsilon) := -\epsilon\log|\epsilon|$.
  % The last two terms can be rewritten as $-{h}(\epsilon) M(P=0) + \epsilon M(P\neq 0)$.

  Using \eqref{eq:loclind} and \eqref{eq:Dbarr1} yields
  \begin{align}
    \Dbarr^{1}(u + \epsilon v) &= \Dbarr(u) - \epsilon \frac{d_{u}'(v)}{d_{u}^{2}} \Dbarr(u)
    \notag\\
    & \qquad + \frac{\epsilon}{d_{u}} \left(\sum_{x:u\neq 0} v(x) \log \frac{|u(x)|}{r(x)} + \sum_{x:u=0} v(x) \log \frac{|v(x)|}{r(x)}\right)
    \notag\\
    & \qquad + \frac{1}{d_{u}} \epsilon\log|\epsilon| v(u=0) + \epsilon v(u\neq 0) + \smo(\epsilon)
    \notag\\
    &= \Dbarr^{1}(u) - \epsilon \frac{d'_{u}(v)}{d_{u}} \Dbarr^{1}(u)
    \notag\\
    & \qquad + \frac{\epsilon}{d_{u}} \left(\sum_{x:u\neq 0} v(x) \log \frac{|u(x)|}{r(x)} + \sum_{x:u=0} v(x) \log \frac{|v(x)|}{r(x)}\right)
    \notag\\
    & \qquad + \frac{1}{d_{u}} \epsilon\log|\epsilon| v(u=0) + \epsilon v(u\neq 0) + \smo(\epsilon).
    \label{eq:Dbar-expansion}
  \end{align}

  Now let $u$ be a local maximizer of $\Dbarr$ in $\ker A$ subject to $d_{u}=1$.  Then $u$ is also a local maximizer of
  $\Dbarr^{1}$ by Remark \ref{sec:reformulatingDbarr}.  Therefore the first statement follows from the facts that the
  derivative of $\epsilon\log\epsilon$ diverges at zero and the coefficient $\frac{1}{d_{u}}v(u = 0)$ changes its sign
  if $v$ is replaced by $-v$.  Since $v(u\neq 0) = v(\Xcal) - v(u=0) = 0$ the inequality follows for all $v\in\ker A$.
  If $\supp(v)\subseteq\supp(u)$ then $d'_{u}(-v) = -v(u>0) = - d'_{u}(v)$.  In this case the left hand side of the
  inequality changes its sign when $v$ is replaced by $-v$, thus it holds as an equality.
\end{proof}
\begin{defi}
  \label{def:critptH}
  A point $u \in\ker A$ is called a \emph{quasi-critical point} of $\Dbarr$ if it satisfies the conditions~\ref{it:Var0}.{}
  and~\ref{it:Var1}.{} of proposition~\ref{prop:critpts}.
\end{defi}
The importance of this definition is that every local extremum of $\Dbarr$ is also a quasi-critical point by the above
proposition. This means that any convergent numerical optimisation algorithm will at least find a quasi-critical point.

\begin{rem}
% \begin{itemize} \item
  Condition~{\ref{eq:Var0}.} of Proposition~\ref{prop:critpts} depends on $u$ only through the support of $u$.
  Therefore it can be used as a necessary condition to test whether a maximizer of $\Dbarr$ can have a given support.
  Since this equation is linear in $v$ it is enough to check it for a basis of $\ker A$.
%   \item Equation \eqref{eq:Var1} can be written as an algebraic equation if $d'_{P}(M) = 0$ after exponentiationg and
%     clearing denominators. %This algebraic equation is homogeneous for example if $P$ has full support.
%   \end{itemize}
\end{rem}
\begin{rem}
  \label{rem:Kandker}
  Condition {\ref{it:Var1}.} of Proposition~\ref{prop:critpts} is also linear in $v$, since $d'_{u}(v) = v(u>0)$ is
  linear in this case.  Moreover, it is trivially satisfied for $v=u$.  This means that it is enough to check condition
  \ref{it:Var1} on a basis of any subspace $K\subset \ker A$ such that the span of $K$ and $u$ contains all $v\in\ker A$
  with $\supp (v) \subseteq \supp (u)$.  A possible choice is
  \begin{equation}
    K^{u} = \{ v\in\ker A: \supp (v) \subseteq \supp (u) \text{ and } d'_{u}(v) = 0 \}.
  \end{equation}
  In this subspace, the equations of proposition \ref{prop:critpts}, {\ref{it:Var1}.} simplify to
  \begin{equation}
    \label{eq:Var1-simple}
    \sum_{x:u\neq 0} v(x) \log \frac{|u(x)|}{r(x)} = 0
  \end{equation}
  for all $v\in K$.
\end{rem}

\section{The codimension one case}
\label{sec:codim1}

In this section the theory developed in the previous sections will be applied to the case where the exponential family
has codimension one.
\begin{ex}
  \label{ex:1}
  If $\ker A$ is onedimensional, then it is spanned by a single vector $u = P_{+} - P_{-}$, where $P_{+}$ and $P_{-}$
  are two probability measures.  If $H_{r}(P_{+}) = H_{r}(P_{-})$, then both $P_{+}$ and $P_{-}$ are global maximizers
  of $\DE$.  Otherwise assume that $H_{r}(P_{+}) < H_{r}(P_{-})$.  Then $P_{+}$ is the global maximizer of $\DE$.  Note
  that $-u$ is another local maximizer of $\Dbarr$.  It is easy to see that $P_{-}$ is also a local maximizer of $\DE$.

  This example can serve as a source of examples and counterexamples.  For example, it is easy to see that for a general
  exponential family, $\supp(P_{\Ecal})$ can be an arbitrary set~$\Ycal$ of cardinality greater or equal to two: Just
  choose two measures $P_{+}$, $P_{-}$ of disjoint support such that $\supp (P_{+})\cup\supp (P_{-}) = \Ycal$, let
  $u=P_{+}-P_{-}$ and choose a matrix $A$ such that $\ker A$ is spanned by $u$.  In the same way one can prove the
  following statements:
  \begin{itemize}
  \item Any set $\Ycal\subsetneq\Xcal$ with cardinality less than $|\Xcal|-1$ is the support of a global maximizer $P$
    of $\DE$ for some exponential family $\Ecal$.
  \item Any measure supported on a set $\Ycal\subsetneq\Xcal$ with cardinality less than $|\Xcal|$ is a local
    maximizer of $\DE$ for some exponential family $\Ecal$.
  \item Any measure supported on a set $\Ycal\subsetneq\Xcal$ with cardinality less than $|\Xcal|-1$ is a global
    maximizer of $\DE$ for some exponential family $\Ecal$.
  \end{itemize}
  Of course, these statements are not true anymore, when the reference measure is fixed or when the class of exponential
  families is restricted in any way.
 % that an arbitrary probability measure can be a local maximizer of a suitable exponential family
 %  (or even a global maximizer, if the support is smaller than $|\Xcal|-2$ and $r$ is adjusted).
\end{ex}
\begin{ex}
  \begin{figure}[ptb]
    \centering
    \includegraphics{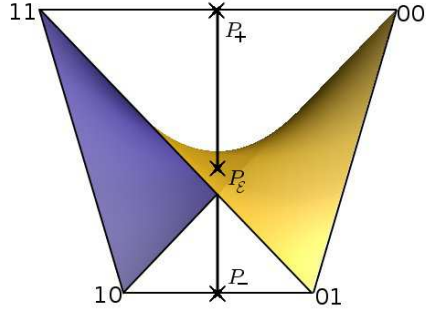}
    \caption{The binary independence model.}
    \label{fig:binind}
  \end{figure}
  As a special case of the previous example, consider the binary independence model with $\Xcal=\{00,01,10,11\}$,
  \begin{equation}
    A = \left(
      \begin{matrix}
        \smash{\overset{00}{1}} & \smash{\overset{01}{1}} & \smash{\overset{10}{0}} & \smash{\overset{11}{0}} \\
        0 & 0 & 1 & 1 \\
        1 & 0 & 1 & 0 \\
        0 & 1 & 0 & 1
      \end{matrix}
    \right),
  \end{equation}
  and $r(x)=1$ for all $x\in\Xcal$.  It is easy to see that $\Ecal$ consists of all probability measures $P$ which
  factorize as $P(x_{1}x_{2}) = P_{1}(x_{1})P_{2}(x_{2})$, justifying the name of this model.  The kernel is spanned by
  \begin{equation}
    u = (+1, -1, -1, +1),
  \end{equation}
  corresponding to two global maximizers $P_{+}=\frac{1}{2}(\delta_{00}+\delta_{11})$ and
  $P_{-}=\frac{1}{2}(\delta_{01}+\delta_{10})$ (see figure \ref{fig:binind}).
\end{ex}

\section{Solving the critical equations}%An algorithm for hierarchical models}
\label{sec:algo}

\newcommand{\Ms}{M^{\sigma}}%
\newcommand{\Msp}{M^{\prime\sigma}}%
Finding the maximizers of $\Dbarr$ has some advantages over directly finding the maximizers of $\DE$, mainly because of
two reasons:
\begin{enumerate}
\item The dimension of the problem is reduced: Instead of maximizing over the whole probability simplex the maximization
  takes place over a convex subset of the kernel of the matrix $A$.  Therefore the dimension of the problem is reduced
  by the dimension of the exponential family.
\item A projection on the exponential family is not needed: $\Dbarr$ can be computed by a ``simple'' formula.
\end{enumerate}
A numerical search for the maximizers using gradient search algorithms is now feasible for larger models. However, there
may be a lot of local maximizers, so it is still a difficult problem to find the global maximizers of $\DE$.  Of course,
the above ideas %can be used not only for the numerical search of maximizers. They
can also be used with symbolic calculations in order to investigate the maximizers of $\DE$.
%, at least for hierarchical exponential families.

In the following assume that the sufficient statistics matrix $A$ has only integer entries.  In this case the $\ker A$
has a basis of integer vectors.  An important class of examples where this condition is satisfied are hierarchical
models. % $M_{1},\dots,M_{k}$.  %This will lead to polynomial critical equations.
%Furthermore in this section $r$ will be the uniform measure $r(x)=1$ in order to keep the notation simple.

Under these assumptions we turn to the equations of Proposition \ref{prop:critpts}.  The main observation is that
equation \eqref{eq:Var1} is algebraic for suitable $u$ once we fix the sign vector of $u$.  This motivates to look
independently at each possible sign vector $\sigma$ that occurs in $\ker A$.
%It turns out that we can say a lot for the sign vectors with full support.

\begin{rem}
  Before investigating the critical equations some short remarks on the sign vectors are necessary.  The set of possible
  sign vectors occuring in a vector space (in this case $\ker A$) forms a (realizable) \emph{oriented matroid}.  A sign
  vector $\sigma$ is called an \emph{(oriented) circuit} if its support $\{x\in\Xcal: \sigma_{x}\neq 0\}$ is inclusion
  minimal.  See the first chapter of \cite{BjoernerLasVergnasSturmfelsWhiteZiegler93:Oriented_Matroids} for an
  introduction to oriented matroids.

  Every sign vector can be written as a \emph{composition} $\sigma_{1}\circ\dots\circ\sigma_{n}$ of circuits, where
  $\circ$ is the associative operation defined by
  \begin{equation}
    (\sigma_{i}\circ\sigma_{i+1})_{x} = 
    \begin{cases}
      (\sigma_{i})_{x} & \text{ if }(\sigma_{i})_{x}\neq 0, \\
      (\sigma_{i+1})_{x} & \text{ else }.
    \end{cases}
  \end{equation}
  There is a free software package \verb|TOPCOM|\cite{TOPCOM} which computes the signed circuits of a matrix.  However,
  this package does not (yet) compute all the sign vectors, but this second step is easy to implement.

  There is a second possible algorithm for computing all sign vectors of an oriented matroid, which shall only be
  sketched here, since it uses the complicated notion of duality (see
  \cite{BjoernerLasVergnasSturmfelsWhiteZiegler93:Oriented_Matroids} for the details): Namely, the set of all sign
  vectors is characterized by the so-called \emph{orthogonality property}, meaning that the set of all sign vectors can
  be computed by calculating all \emph{cocircuits} and checking the orthogonality property on each possible vector
  $\sigma\in\{0,\pm1\}^{\Xcal}$.

  The nonzero sign vectors $\sigma$ occuring in a vector space always come in pairs $\pm\sigma$.  It is customary to
  list only one representative of each such pair.  This is not a problem, since the function $\Dbarr$ is antisymmetric,
  i.e., a local maximizer $u$ with sign vector $\sgn(u)=-\sigma$ corresponds to a local minimizer $-u$ with sign vector
  $\sigma$, and both will be quasi-critical points of $\Dbarr$.
\end{rem}

Now fix a sign vector $\sigma$ %with full support
and choose $u_{0}\in \ker A$ such that $\sgn(u_{0}) = \sigma$ and $d_{u_{0}}=1$.  Denote $\Ycal := \supp(\sigma) =
\supp(u_{0})$.
% We can assume that $\sigma$ does not contradict statement \eqref{eq:Var0} of Proposition \ref{prop:critpts}.
Define $d^{\sigma}(v) := \sum_{x:\sigma_x>0} v(x)$.  This implies $d^{\sigma}(v) = d'_{u_{0}}(v)$ whenever
$\supp(v)\subseteq \supp(u_{0}) = \supp(\sigma)$.  Let
\begin{equation}
  K^{\sigma} := \big\{ v\in\ker A : d^{\sigma}(v) = 0\text{ and }\supp(v)\subseteq\supp(\sigma)\big\}.%, \text{ and } v(x)=0 \text{ for all } x \text{ such that } \sigma_{x}= 0 
\end{equation}
If $u\in\ker A$ satisfies $d_{u}=1$ and $\sgn(u)=\sigma$, then $u-u_{0}\in K^{\sigma}$.
%
% We can choose vectors $\Msp_{1},\dots,\Msp_{k-1}\in\ker_{\Zb} A$ such that $P_{0}, \Msp_{1}, \dots, \Msp_{k-1}$ is a
% basis of $\ker A$ and $d'_{P_{0}}(\Msp_{i}) = d^{\sigma}(\Msp_{i}) = 0$. This can be achieved similarly to the
% Gram-Schmidt method: Since $P$ is a linear combination of the $M_{i}$ we can assume $d(M_{k}) > 0$ without loss of
% generality. Then let $\Msp_{i} := d^{\sigma}(M_{k}) M_{i} - d^{\sigma}(M_{i}) M_{k} \in\ker_{\Zb}A$ for $1\le i <
% k$. Obviously the $\Msp_{i}$ are linearly independent, and since $d^{\sigma}(P_{0}) = 1 \neq 0 = d^{\sigma}(P_{0})$ the
% set $\{P_{0}, \Msp_{1}, \dots, \Msp_{k-1}\}$ is linearly independent and thus a basis.
%
% Any point $P\in \ker A$ with $\sgn(P) = \sigma$ and $d(P)=1$ can be written as $P = P_{0} + \sum_{i=1}^{k-1}\tlambda_{i}
%\Msp_{i}$.
By definition %\ref{def:critptH}
$u$ is a quasi-critical point of $\Dbarr$ if and only if
\begin{equation}
%  \label{eq:tMlogabsP}
  \sum_{x\in\Ycal} v(x) \log\frac{|u(x)|}{r(x)} %+ \sum_{x\notin\Ycal} v(x) \log|v(x)| 
  = 0 \text{ for all } v\in K^{\sigma}
\end{equation}
(see Remark \ref{rem:Kandker}).  These equations are linear in $v$, so it is enough to consider them for a spanning set
of $K^{\sigma}$.  Since by assumption the matrix $A$ has only integer entries the set
\begin{equation}
  \label{eq:KsigmaZ}
  K^{\sigma}_{\Zb} := K^{\sigma}\cap\Zb^{\Xcal}
\end{equation}
contains a spanning set of $K^{\sigma}$.  Therefore $u$ is a quasi-critical point of $\Dbarr$ if and only if
\begin{equation}
%  \label{eq:tMlogabsP}
  \sum_{x\in\Ycal} v(x) \log\frac{|u(x)|}{r(x)} %+ \sum_{x\notin\Ycal} v(x) \log|v(x)| 
  = 0 \text{ for all } v\in K^{\sigma}_{\Zb}.
\end{equation}
% Since $|P(x)| = \sigma_x P(x)$ we can write
% \begin{equation}
%   \log |P(x)| = \log\left(\sigma_x (P_{0}(x) + \sum_{i}\tlambda_{i}\Msp_{i}(x))\right).
% \end{equation}
Exponentiating these equations gives
\begin{equation}
  \label{eq:critpolyt}
  \prod_{x\in\Ycal:v(x)>0} \left(\frac{\sigma_x u(x)}{r(x)}\right)^{v(x)}
%\\
  = %\prod_{x\notin\Ycal} |v(x)|^{- v(x)}
  \prod_{x\in\Ycal:v(x)<0} \left(\frac{\sigma_x u(x)}{r(x)}\right)^{-v(x)} \text{ for all } v\in K^{\sigma}_{\Zb}.
\end{equation}
This is a system of \emph{polynomial} equations.  Every solution $u\in u_{0} + K^{\sigma}$ to this system that
satisfies %$\sgn\left(P_{0} + \sum_{i=1}^{k-1}\tlambda_{i}\Msp_{i}\right) = \sigma$
$\sgn(u) = \sigma$ is a quasi-critical point of $\Dbarr$ and thus a potential maximizer.
%
%In this algebraic system one can replace the parametrization $P = P_{0} + \sum_{i=1}^{k-1}\tlambda_{i} \Msp_{i}$ by the
%parametrization $P = \sum_{i=1}^{k}\lambda_{i} M_{i}$. In this case one has to add the additional linear equation $d(P)
%= 1$. This means that we do not need to know $P_{0}$.
%
% If $\supp(\sigma)=\Ycal\neq\Xcal$, then the support constraint $\supp(P) = \sigma$ contains one equation for every zero
% in $\sigma$. These equations can be incorporated into the parametrization if we replace $\{M_{1},\dots,M_{k}\}$ by a
% basis $\{\Ms_{1},\dots,\Ms_{l}\}$ of $\ker A \cap \Rb^{\Ycal}$. Then we have a parametrization
% $P(\lambda_{1},\dots,\lambda_{l}) = \sum_{i=1}^{l}\lambda_{i}\Ms_{i}$ where the support of $P$ is contained in the
% support of $\sigma$ (i.e. we don't have to worry about supports that are too large).
% At this point we can do one further simplification.  By proposition \ref{prop:critpts} we have
% \begin{equation}
%   \sum_{x}
% \end{equation}

At this point it is possible to do one more simplification: If $v\in K_{\Zb}^{\sigma}$, then $v(\sigma<0)=v(\sigma\neq0)
- v(\sigma>0) = 0$.  It follows that $v_{+}(\sigma<0) + v_{-}(\sigma<0) = 0$, so
\begin{equation}
  \prod_{x:v(x)>0}(\sigma_{x})^{v(x)} = (-1)^{v_{+}(\sigma<0)} = (-1)^{v_{-}(\sigma<0)} = \prod_{x:v(x)<0}(\sigma_{x})^{-v(x)}
\end{equation}
All in all this yields:
\begin{prop}
  \label{prop:critpoly}
  Fix a sign vector $\sigma\in\{\pm 1\}^{\Xcal}$. %with full support
  Let $u\in\ker A$ satisfy $d_{u} = 1$ and
  \begin{equation}
  \label{eq:critpoly}
    \prod_{x\in\Ycal} \left(\frac{u(x)}{r(x)}\right)^{v_{+}(x)}
    = %\prod_{x\notin\Ycal} |v(x)|^{- v(x)}
    \prod_{x\in\Ycal} \left(\frac{u(x)}{r(x)}\right)^{v_{-}(x)}
  \end{equation}
  for all $v = v_{+}-v_{-}\in K_{\Zb}^{\sigma}$.  If $\sgn\left(u\right) = \sigma$, then $u$ is a quasi-critical point
  of $\Dbarr$. Every quasi-critical point of $\Dbarr$ %with full support
  arises in this way.
\end{prop}
\begin{rem}
  \label{rem:choiceofequations}
  Note that the system of equations \eqref{eq:critpoly} still contains infinitely many equations.  The argument before
  equation \eqref{eq:KsigmaZ} shows that a finite number of equations is enough.  However, there are different possible
  choices for this finite set (at least a basis of $K_{\Zb}^{\sigma}$ is needed), and the choice may have a large
  computational impact.  This issue will be addressed below.
\end{rem}

Proposition \ref{prop:critpoly} shows that the maximizers of $\Dbarr$ can be found by analyzing all the solutions to the
algebraic systems of equations \eqref{eq:critpoly} for all different possible sign vectors~$\sigma$.  Since the analysis
of systems of polynomials works best over the complex numbers, in the following these equations will be considered as
complex equations in the variables~$u(x)$.  Of course, only real solutions with the right sign pattern will be candidate
solutions of the original maximization problem.

From now on fix $\sigma$ again.  Define $I_{2}^{\sigma} \subseteq \Cb[u(x) : x\in\Ycal]$ to be the
\emph{ideal}\footnote{The mathematicel disciplines of studying polynomial equations and their solution sets are
  commutative algebra and algebraic geometry.  In the following some definitions from these two fields are used.  The
  reader is refered to \cite{CoxLittleOShea08:Ideals_Varieties_Algorithms} for exact definitions and the basic facts.}
generated by
%\begin{itemize}
%\item
all equations \eqref{eq:critpoly} in the polynomial ring $\Cb[u(x) : x\in\Ycal]$ with one variable for each $x\in\Ycal$.
%\item the equation $d(P)-1 = 0$,
%\item and all equations $P(x) = 0$ for $x\notin\Ycal$.
%\end{itemize}
Similarly, let $I_{1}^{\sigma} \subseteq \Cb[u(x) : x\in\Ycal]$ be the ideal generated by the equations
\begin{equation}
  \sum_{x\in\Ycal}A_{i,x}u(x) = 0,\qquad\text{ for all }i.
\end{equation}
Finally let $I^{\sigma}:= I_{1}^{\sigma}+I_{2}^{\sigma}$.  The set of all common complex solutions of all equations in
$I^{\sigma}$ is an algebraic subvariety of $\Cb^{\Ycal}$ and will be denoted by $X^{\sigma}$.

\begin{rem}
  \label{rem:projvar}
  Note that we omitted the equation $d_{u}-1 = 0$ in the definition of the ideal.  It is easy to see that we can ignore
  this condition at first, because every solution satisfying $\sgn(u) = \sigma$ has $d(u)\neq 0$ and can thus be
  normalized to a solution with $d_{u} = 1$.  In other words, the original problem is solved once all points on the
  variety $X^{\sigma}$ that satisfy the sign condition are known.  The algebraic reason for this fact is that all the
  defining equations of $I$ are homogeneous.  This means that we can also replace $X^{\sigma}$ by the projective variety
  corresponding to $I^{\sigma}$, which is another interpretation of the fact that the normalization does not matter at
  this point.
\end{rem}
Both ideals $I_{1}^{\sigma}$ and $I_{2}^{\sigma}$ taken for themselves are very nice: $I_{1}^{\sigma}$ corresponds to a
system of linear equations, so it can be treated by the methods of linear algebra.  On the other hand, $I_{2}^{\sigma}$
is a system of binomial equations, and there are a lot of theoretical results and fast algorithms for binomial
equations\cite{EisenbudSturmfels96:Binomial_Ideals,Kahle10:Decompositions_of_binomial_ideals}.  However, the sum of a
linear ideal and a binomial ideal can be arbitrarily complicated.  In fact, it is easy to see that any ideal can be
reparameterized as a sum of a linear ideal and a binomial ideal: For example, a polynomial equation $\sum_{i}m_{i}=0$,
where $m_{i}$ are arbitrary monomials, is equivalent to the system of equations
\begin{gather*}
  z_{i} - m_{i} = 0, \text{ for all }i,\\
  \sum_{i}z_{i} = 0,
\end{gather*}
where one additional variable $z_{i}$ has been introduced for every monomial.  Still, the two ideals $I_{1}^{\sigma}$
and $I_{2}^{\sigma}$ under consideration here are closely related, so there is hope that general statements can be made.

$X^{\sigma}$ equals the intersection of $X_{1}^{\sigma}$ and $X_{2}^{\sigma}$, where $X_{1}^{\sigma}$ and
$X_{2}^{\sigma}$ are the varieties of $I_{1}^{\sigma}$ and $X_{2}^{\sigma}$ respectively.  The variety $X_{1}^{\sigma}$
is easy to determine: By definition it is given by the (complex) kernel of $A$ restricted to $\Ycal$:
\begin{equation}
  X_{1}^{\sigma}=\ker_{\Cb} A \cap\Cb^{\Ycal}.
\end{equation}
The variety $X_{2}^{\sigma}$ is a little bit more complicated, but still a lot can be said.

By definition, $I_{2}^{\sigma}$ is generated by a countable collection of binomials.  In fact, Hilbert's Basissatz shows
that a finite subset of the generators of $I_{2}^{\sigma}$ is sufficient to generate the ideal.  In general it can be a
difficult task to find such a finite subset, but since equations \eqref{eq:critpoly} correspond to directional
derivatives, it is sufficient to consider them for any basis $B$ of $K^{\sigma}$ (see Remark
\ref{rem:choiceofequations}).  So denote the ideal generated by the equations corresponding to a basis $B$ of $\ker A$
by $I_{2}(B)$.  In general $I_{2}(B)$ will have a different solution variety $V(I_{2}(B))$ than $I_{2}^{\sigma}$, and
moreover $V(I_{2}(B))$ will depend on $B$.  From what was said above all these varieties agree on the orthant of
$\Rb^{\Xcal}$ defined by $\sgn = \sigma$.  The presence of additional (complex) solutions outside this orthant may
complicate the algebraic analysis.  It is obvious that all the ideals $I_{2}(B)$ are contained in $I_{2}^{\sigma}$.
This means that $I_{2}^{\sigma}$ has the smallest solution set, so a finite generating set of $I_{2}^{\sigma}$ would be
useful.

More precisely, since the ideal $I_{2}^{\sigma}$ is generated by binomials, the theory of
\cite{EisenbudSturmfels96:Binomial_Ideals} applies.  Corollary 2.6 of this work implies that the ideal $I_{2}^{\sigma}$
is a prime ideal.
% (the corresponding saturated partial character is given by $\rho = \exp(H\big|_{\Ycal})$, cf. corollary \ref{cor:Hadd}).
This means that $X_{2}^{\sigma}$ is \emph{irreducible}, i.e., it can not be written as a union of two proper
subvarieties.  Binomial prime ideals are also called \emph{toric ideals}\cite[remark before Corollary
2.6]{EisenbudSturmfels96:Binomial_Ideals}.  However, it is easy to construct examples such that $I_{2}(B)$ is not
irreducible.

Fortunately there are fast computer algorithms, implemented in the software package \verb|4ti2|\cite{4ti2}, which can be
used to compute a finite generating set of $I_{2}^{\sigma}$ \cite{HemmeckeMalkin09:Generating_Sets_Lattice_Ideals}.
These algorithms compute finite generating sets of so-called \emph{lattice ideals}.  It turns out that $I_{2}^{\sigma}$
becomes a lattice ideal after a rescaling of the coordinates.  To be concrete, writing $u_{r}(x) := \frac{u(x)}{r(x)}$
yields a new, equivalent ideal $I_{2,r}^{\sigma}\subseteq\Cb[u_{r}(x):x\in\Ycal]$ generated by the binomials
\begin{equation}
  \prod_{x:v>0} u_{r}(x)^{v(x)} = \prod_{x:v<0} u_{r}(x)^{-v(x)},\quad\text{ for all }v\in K^{\sigma}_{\Zb}.
\end{equation}
The ideal $I_{2,r}^{\sigma}$ is called a lattice ideal, since it is related to the integer lattice
$K^{\sigma}_{\Zb}\subseteq\Zb^{\Ycal}$.

% The solution set $X_{\sigma}$ of the system \eqref{eq:critpoly} is an algebraic
% variety. $X_{\sigma}$ is a subset of $\Cb^{l}$, each point corresponding to an $l$-tuple
% $(\lambda_{1},\dots,\lambda_{l})$. By the linear map $(\lambda_{i})_{i}\mapsto \sum_{i=1}^{l}\lambda_{i}\Ms_{i}$ it can
% also be considered as a subset of $\ker_{\Cb} A$.

% Otherwise, it is possible to compute $I_{M}$ via $I_{M} = I_{M}(B) : (\prod_{x}M(x))^{\infty}$ with any dedicated
% computer algebra system (e.g., Macaulay2\cite{M2} or \textsc{Singular}\cite{Singular})

Now we turn to $X^{\sigma} = X_{1}^{\sigma} \cap X_{2}^{\sigma}$.  Even though $X_{1}^{\sigma}$ and $X_{2}^{\sigma}$ are
irreducible, in general $X^{\sigma}$ will be reducible.
%A variety $X$ is called \emph{reducible} if it can be written as a nontrivial union $X=U\cup V$ of varietes, i.e. $X\neq U$ and $X\neq V$.
This means that we can write $X^{\sigma}$ as a finite union of irreducible components $X^{\sigma} =
V^{\sigma}_{1}\cup\dots\cup V^{\sigma}_{c}$.  To each of these components $V^{\sigma}_{i}$ corresponds a polynomial
ideal $I_{i}^{\sigma}$, and we have $u\in X^{\sigma}$ if and only if $u$ solves (at least) one of these ideals.  The
procedure to obtain the ideals $I_{i}^{\sigma}$
%that describe the different components $V_{i}^{\sigma}$ 
is called \emph{primary decomposition}.
%There are symbolic algorithms as well as numerical algorithms in order to tackle this problem.
%The first step to solve the system \eqref{eq:critpoly} is a \emph{primary decomposition}. 

If an irreducible component $V^{\sigma}_{i}$ is zero-dimensional, then it consists of only one point, and it is easy to
check whether this unique element $u\in X^{\sigma}_{i}$ satisfies $\sgn(u) = \sigma$.  However, components of positive
dimension may arise.  In this case it is not easy to see whether these components contain elements $u$ satisfying
$\sgn(u) = \sigma$. Fortunately, in many cases this information is not required:

\begin{thm}
  \label{thm:Dsigmaind}
%  Let $u, u'$ be two elements of an irreducible component $V$ of $X^{\sigma}$ such that $u(x)u'(x)\neq 0$ for
  Let $u$ be an element of an irreducible component $V$ of $X^{\sigma}$ such that $d_{\sigma}(u)=1$.  Suppose there
  exists $u_{0}\in V$ such that $d_{\sigma}(u_{0})=1$ and $\sgn(u_{0})=\sigma$.
%Fix a logarithm $\Cb\setminus\{0\}\to\Cb$.
  Then
  \begin{equation}
    \Dbarr(u_{0}) = \sum_{x\in\Ycal} \Re(u(x))\log\frac{|u(x)|}{r(x)}.
  \end{equation}
%   % and define
%   and define the function
% %  Then the real part of the function
%   \begin{align*}
%     \Dbar_{\sigma}: (\Cb\setminus\{0\})^{\Ycal} &\to \Cb
%     \\
%     v \mapsto \sum_{x} v(x)\log(\sigma_{x}v(x)).% - \sum_{x}\Im(v(x)) \frac{1}{i}\sum_{x}
%   \end{align*}
%   Then $\Re(\Dbar_{\sigma}(v)) - \sum_{x}\Im(v(x))$
\end{thm}
\begin{proof}
%  First assume that $u$ and $u'$ are both regular points of $V$.
  Let
  \begin{equation}
    V':=\{v\in V:v(x)\neq 0\text{ for all }x\in\Ycal\text{, and }d_{\sigma}(v)\neq 0\}. % v\text{ is regular, and }
  \end{equation}
  Then $V'$ is a Zariski-open subset of $V$, hence $V'$ is irreducible.  This implies that $V'$ is pathconnected,
%.  Since  $V'$ is regular, $V'$ is a smooth complex manifold, 
  so there exists a smooth path $\gamma:[0,1]\to V'$ from $u$ to $u_{0}$.  This is obvious if $V'$ is regular, since
  then $V'$ is a locally pathconnected and connected complex manifold.  It follows that all regular points can be
  connected by a smooth path.  Finally, every singular point $p$ can be linked by a smooth path to some regular point in
  any neighbourhood of $p$.  By Remark \ref{rem:projvar} this path can be chosen such that $d_{\sigma}(\gamma_{t})=1$
  for all $t\in[0,1]$.

  Fix a point $u\in V'$ and fix a convention for the logarithm.  For every $x\in\Ycal$ the logarihm can be continued to
  a map $t\mapsto \log^{t,x}(\gamma_{t}(x))$.  For every $t\in[0,1]$ define a linear functional $s_{t}:
  K^{\sigma}_{\Cb}\to\Cb$ via 
  \begin{equation}
    s_{t}(v) = \frac{1}{2\pi i} \sum_{x\in\Ycal}v(x)\log^{t,x}\frac{\sigma_{x}\gamma_{t}(x)}{r(x)}.
  \end{equation}
  By definition of $X^{\sigma}$ it follows that $s_{t}$ takes only integer values on $K^{\sigma}_{\Zb}$, and $s_{t}$ can
  be identified with an element of the dual lattice $K^{\sigma*}_{\Zb}$ of $K_{\Zb}^{\sigma}$.  Since
  $K_{\Zb}^{\sigma*}$ is a discrete subset of the dual vector space $K_{\Cb}^{\sigma*}$ and since the map $t\mapsto
  s_{t}$ is continuous $s_{t}$ is constant along $\gamma$.

  Now consider the function $f(t) =
  \sum_{x\in\Ycal}\gamma_{t}(x)\log^{t,x}\left(\frac{\sigma_{x}\gamma_{t}(x)}{r(x)}\right)$.  Its derivative is $f'(t) =
  \sum_{x\in\Ycal} \gamma'_{t}(x)\log^{t,x}\left(\frac{\sigma_{x}\gamma_{t}(x)}{r(x)}\right) = 2\pi i
  s_{0}(\gamma'_{t})$, where $\gamma'_{t}(x)=\pd{t}\gamma_{t}(x)\in K_{\Zb}^{\sigma}$.
  %by construction.
%Thus $f'(t) =  2\pi i s_{0}(\gamma_{t})$, and 
  It follows that $f(1) - f(0) = 2\pi i s_{0}(\gamma_{1} - \gamma_{0})$.  In other words,
  \begin{equation}
    \sum_{x\in\Ycal} u_{0}(x)\log^{1,x}\frac{\sigma_{x}u_{0}(x)}{r(x)} = \sum_{x\in\Ycal} u(x)\log \frac{\sigma_{x}u(x)}{r(x)} + 2\pi i s_{0}(u_{0} - u).
  \end{equation}
  If $\log^{1,x}(\sigma_{x}u_{0}(x)) = \log^{x}(\sigma_{x}u_{0}(x)) + 2\pi i k_{x}$ with $k_{x}\in\Zb$, then
  \begin{equation}
    \Dbarr(u_{0}) = f(0) + 2\pi i \left(s_{0}(u_{0} - u) - \sum_{x\in\Ycal} u_{0}(x) k_{x}\right).
  \end{equation}
  Taking the real parts of this equation gives
  \begin{align*}
    \Dbarr(u_{0})
    %& = \Re(f(0)) - 2\pi \Im(s_{0}(u_{0} - u))
    %\\
    & = \Re(f(0)) + 2\pi s_{0}(\Im(u))
    \\ & = \Re(f(0)) - i \sum_{x\in\Ycal} \Im(u(x))\log\frac{\sigma_{x}u(x)}{r(x)}
    \\ & = \Re\left(\sum_{x\in\Ycal} \Re(u(x))\log\frac{\sigma_{x}u(x)}{r(x)}\right)
    \\ & = \sum_{x\in\Ycal} \Re(u(x))\log\frac{|u(x)|}{r(x)}.
  \end{align*}
  By continuity this formula continues to hold on the closure of $V'$, which equals $V$.
\end{proof}

The theorem implies that in many cases only one point $u$ from each irreducible component of $X^{\sigma}$ needs to be
tested.  Only if $\sum_{x\in\Ycal}\Re(u(x))\log\frac{|u(x)|}{r(x)}$ is exceptionally large it is necessary to analyze
this irreducible component further and see if there is a real point $u_{0}$ from the same irreducible component that
satisfies the sign condition.

\begin{rem}
  The above theorem also makes it possible to use methods of numerical algebraic
  geometry\cite{SommeseWampler:Numerical_Solution_Polynomials}.  These methods can determine the number of irreducible
  components and their dimensions.  Additionally it is possible to sample points from any irreducible component.  In
  fact, each component is represented by a so-called \emph{witness set}, a set of elements of this component.  These
  points can then be used to numerically evaluate $\Dbarr$.  One implementation, available on the Internet, is
  Bertini\cite{Bertini}.
\end{rem}

Let $x\in\Ycal$.  For every irreducible component $X^{\sigma}_{i}$ there are the following alternatives:
\begin{itemize}
\item Either $u(x)=0$ for all $u\in X^{\sigma}_{i}$.  In this case $\sgn(u)\neq \sigma$ on $X^{\sigma}_{i}$.
\item Or $u(x) = 0$ holds only on a subset of measure zero.
\end{itemize}
The reason for this is that the equation $u(x)=0$ defines a closed subset of $X^{\sigma}_{i}$, and either this closed
subset is all of $X^{\sigma}_{i}$, or it has codimension one (this argument needs the irreducibility of
$X^{\sigma}_{i}$).

When computing the primary decomposition the irreducible components of the first kind can be excluded algebraically by a
method called saturation: Namely, the variety corresponding to the saturation
\begin{equation}
  \left(I^{\sigma}:(\prod_{x\in\Ycal}u(x))^{\infty}\right)
  = \Big\{ f\in\Cb[ u(x)]: f m\in I^{\sigma} \text{ for some monomial }m\in\Cb[ u(x)]\Big\}
\end{equation}
consists only of those irreducible components of $X^{\sigma}$ which are not contained in any coordinate plane.  In the
same way we may also saturate by the polynomial $d^{\sigma}(M)$, since any solution $M$ with $\sgn(M)=\sigma$ will have
$0 \neq d(M) = d^{\sigma}(M)$.

The main reason why saturation is important is that it may reduce the complexity of symbolic calculations.

\begin{ex}
  \label{ex:4-2}
  The above ideas can be applied to the hierarchical model (see \cite{Lauritzen96:Graphical_Models}) of pair
  interactions among four binary random variables (the ``binary 4-2 model'').  This exponential family consists of all
  probability distributions of full support which factor as a product of functions that depend on only two of the four
  random variables.

  The maximization problem of this model is related to orthogonal latin squares: If the binary random variables are
  replaced by random variables of size $k$, then the maximizer of the corresponding 4-2 model is easy to find if two
  orthogonal latin squares of size $k$ exist\cite{Matus09:Divergence_from_factorizable_distributions}, and in this case
  the maximum value of $\DE$ equals $2 \log(k)$.  From this point of view, the following discussion will give an
  extremely complicated proof of the trivial fact that there are no two orthogonal latin squares of size two.

  The sufficient statistics may be chosen as
  \begin{equation*}
    A_{4-2}=
    \left(
      \begin{matrix}
        1 \phantom{+} 0 \phantom{+} 0 \phantom{+} 0 \phantom{+} 1 \phantom{+} 0 \phantom{+} 0 \phantom{+} 0 \phantom{+} 1 \phantom{+} 0 \phantom{+} 0 \phantom{+} 0 \phantom{+} 1 \phantom{+} 0 \phantom{+} 0 \phantom{+} 0\\
        0 \phantom{+} 1 \phantom{+} 0 \phantom{+} 0 \phantom{+} 0 \phantom{+} 1 \phantom{+} 0 \phantom{+} 0 \phantom{+} 0 \phantom{+} 1 \phantom{+} 0 \phantom{+} 0 \phantom{+} 0 \phantom{+} 1 \phantom{+} 0 \phantom{+} 0\\
        0 \phantom{+} 0 \phantom{+} 1 \phantom{+} 0 \phantom{+} 0 \phantom{+} 0 \phantom{+} 1 \phantom{+} 0 \phantom{+} 0 \phantom{+} 0 \phantom{+} 1 \phantom{+} 0 \phantom{+} 0 \phantom{+} 0 \phantom{+} 1 \phantom{+} 0\\
%        0 \phantom{+} 0 \phantom{+} 0 \phantom{+} 1 \phantom{+} 0 \phantom{+} 0 \phantom{+} 0 \phantom{+} 1 \phantom{+} 0 \phantom{+} 0 \phantom{+} 0 \phantom{+} 1 \phantom{+} 0 \phantom{+} 0 \phantom{+} 0 \phantom{+} 1\\
        1 \phantom{+} 0 \phantom{+} 1 \phantom{+} 0 \phantom{+} 0 \phantom{+} 0 \phantom{+} 0 \phantom{+} 0 \phantom{+} 1 \phantom{+} 0 \phantom{+} 1 \phantom{+} 0 \phantom{+} 0 \phantom{+} 0 \phantom{+} 0 \phantom{+} 0\\
        0 \phantom{+} 1 \phantom{+} 0 \phantom{+} 1 \phantom{+} 0 \phantom{+} 0 \phantom{+} 0 \phantom{+} 0 \phantom{+} 0 \phantom{+} 1 \phantom{+} 0 \phantom{+} 1 \phantom{+} 0 \phantom{+} 0 \phantom{+} 0 \phantom{+} 0\\
        0 \phantom{+} 0 \phantom{+} 0 \phantom{+} 0 \phantom{+} 1 \phantom{+} 0 \phantom{+} 1 \phantom{+} 0 \phantom{+} 0 \phantom{+} 0 \phantom{+} 0 \phantom{+} 0 \phantom{+} 1 \phantom{+} 0 \phantom{+} 1 \phantom{+} 0\\
%        0 \phantom{+} 0 \phantom{+} 0 \phantom{+} 0 \phantom{+} 0 \phantom{+} 1 \phantom{+} 0 \phantom{+} 1 \phantom{+} 0 \phantom{+} 0 \phantom{+} 0 \phantom{+} 0 \phantom{+} 0 \phantom{+} 1 \phantom{+} 0 \phantom{+} 1\\
        1 \phantom{+} 0 \phantom{+} 1 \phantom{+} 0 \phantom{+} 1 \phantom{+} 0 \phantom{+} 1 \phantom{+} 0 \phantom{+} 0 \phantom{+} 0 \phantom{+} 0 \phantom{+} 0 \phantom{+} 0 \phantom{+} 0 \phantom{+} 0 \phantom{+} 0\\
        0 \phantom{+} 1 \phantom{+} 0 \phantom{+} 1 \phantom{+} 0 \phantom{+} 1 \phantom{+} 0 \phantom{+} 1 \phantom{+} 0 \phantom{+} 0 \phantom{+} 0 \phantom{+} 0 \phantom{+} 0 \phantom{+} 0 \phantom{+} 0 \phantom{+} 0\\
        0 \phantom{+} 0 \phantom{+} 0 \phantom{+} 0 \phantom{+} 0 \phantom{+} 0 \phantom{+} 0 \phantom{+} 0 \phantom{+} 1 \phantom{+} 0 \phantom{+} 1 \phantom{+} 0 \phantom{+} 1 \phantom{+} 0 \phantom{+} 1 \phantom{+} 0\\
%        0 \phantom{+} 0 \phantom{+} 0 \phantom{+} 0 \phantom{+} 0 \phantom{+} 0 \phantom{+} 0 \phantom{+} 0 \phantom{+} 0 \phantom{+} 1 \phantom{+} 0 \phantom{+} 1 \phantom{+} 0 \phantom{+} 1 \phantom{+} 0 \phantom{+} 1\\
        1 \phantom{+} 1 \phantom{+} 0 \phantom{+} 0 \phantom{+} 0 \phantom{+} 0 \phantom{+} 0 \phantom{+} 0 \phantom{+} 1 \phantom{+} 1 \phantom{+} 0 \phantom{+} 0 \phantom{+} 0 \phantom{+} 0 \phantom{+} 0 \phantom{+} 0\\
        0 \phantom{+} 0 \phantom{+} 1 \phantom{+} 1 \phantom{+} 0 \phantom{+} 0 \phantom{+} 0 \phantom{+} 0 \phantom{+} 0 \phantom{+} 0 \phantom{+} 1 \phantom{+} 1 \phantom{+} 0 \phantom{+} 0 \phantom{+} 0 \phantom{+} 0\\
        0 \phantom{+} 0 \phantom{+} 0 \phantom{+} 0 \phantom{+} 1 \phantom{+} 1 \phantom{+} 0 \phantom{+} 0 \phantom{+} 0 \phantom{+} 0 \phantom{+} 0 \phantom{+} 0 \phantom{+} 1 \phantom{+} 1 \phantom{+} 0 \phantom{+} 0\\
%        0 \phantom{+} 0 \phantom{+} 0 \phantom{+} 0 \phantom{+} 0 \phantom{+} 0 \phantom{+} 1 \phantom{+} 1 \phantom{+} 0 \phantom{+} 0 \phantom{+} 0 \phantom{+} 0 \phantom{+} 0 \phantom{+} 0 \phantom{+} 1 \phantom{+} 1\\
        1 \phantom{+} 1 \phantom{+} 0 \phantom{+} 0 \phantom{+} 1 \phantom{+} 1 \phantom{+} 0 \phantom{+} 0 \phantom{+} 0 \phantom{+} 0 \phantom{+} 0 \phantom{+} 0 \phantom{+} 0 \phantom{+} 0 \phantom{+} 0 \phantom{+} 0\\
        0 \phantom{+} 0 \phantom{+} 1 \phantom{+} 1 \phantom{+} 0 \phantom{+} 0 \phantom{+} 1 \phantom{+} 1 \phantom{+} 0 \phantom{+} 0 \phantom{+} 0 \phantom{+} 0 \phantom{+} 0 \phantom{+} 0 \phantom{+} 0 \phantom{+} 0\\
        0 \phantom{+} 0 \phantom{+} 0 \phantom{+} 0 \phantom{+} 0 \phantom{+} 0 \phantom{+} 0 \phantom{+} 0 \phantom{+} 1 \phantom{+} 1 \phantom{+} 0 \phantom{+} 0 \phantom{+} 1 \phantom{+} 1 \phantom{+} 0 \phantom{+} 0\\
%        0 \phantom{+} 0 \phantom{+} 0 \phantom{+} 0 \phantom{+} 0 \phantom{+} 0 \phantom{+} 0 \phantom{+} 0 \phantom{+} 0 \phantom{+} 0 \phantom{+} 1 \phantom{+} 1 \phantom{+} 0 \phantom{+} 0 \phantom{+} 1 \phantom{+} 1\\
        1 \phantom{+} 1 \phantom{+} 1 \phantom{+} 1 \phantom{+} 0 \phantom{+} 0 \phantom{+} 0 \phantom{+} 0 \phantom{+} 0 \phantom{+} 0 \phantom{+} 0 \phantom{+} 0 \phantom{+} 0 \phantom{+} 0 \phantom{+} 0 \phantom{+} 0\\
        0 \phantom{+} 0 \phantom{+} 0 \phantom{+} 0 \phantom{+} 1 \phantom{+} 1 \phantom{+} 1 \phantom{+} 1 \phantom{+} 0 \phantom{+} 0 \phantom{+} 0 \phantom{+} 0 \phantom{+} 0 \phantom{+} 0 \phantom{+} 0 \phantom{+} 0\\
        0 \phantom{+} 0 \phantom{+} 0 \phantom{+} 0 \phantom{+} 0 \phantom{+} 0 \phantom{+} 0 \phantom{+} 0 \phantom{+} 1 \phantom{+} 1 \phantom{+} 1 \phantom{+} 1 \phantom{+} 0 \phantom{+} 0 \phantom{+} 0 \phantom{+} 0\\
%        0 \phantom{+} 0 \phantom{+} 0 \phantom{+} 0 \phantom{+} 0 \phantom{+} 0 \phantom{+} 0 \phantom{+} 0 \phantom{+} 0 \phantom{+} 0 \phantom{+} 0 \phantom{+} 0 \phantom{+} 1 \phantom{+} 1 \phantom{+} 1 \phantom{+} 1
        1 \phantom{+} 1 \phantom{+} 1 \phantom{+} 1 \phantom{+} 1 \phantom{+} 1 \phantom{+} 1 \phantom{+} 1 \phantom{+} 1 \phantom{+} 1 \phantom{+} 1 \phantom{+} 1 \phantom{+} 1 \phantom{+} 1 \phantom{+} 1 \phantom{+} 1
      \end{matrix}
    \right)
  \end{equation*}
  Here, the columns are ordered in such a way that the column number $i+1$ corresponds to the state
  $x_{i}\in\Xcal=\{0,1\}^{4}$ that is indexed by the binary representation of $i\in\{0,\dots,15\}$.

  The software package \verb|TOPCOM|\cite{TOPCOM} is used to calculate the oriented circuits of $\ker A$, from which all
  sign vectors are computed by composition.  Up to symmetry there are 73 different sign vectors occuring in $\ker A$.
  Here, the symmetry of the model is generated by the permutations of the four binary units and the relabelings
  $0\leftrightarrow 1$ of each unit.

  From these 73 sign vectors only 20 satisfy condition~{\ref{it:Var0}.} of Proposition~\ref{prop:critpts}.  The sign
  vectors of small support are easy to handle: There are two sign vectors $\sigma_{1},\sigma_{2}$ whose support has
  cardinality eight.  They are in fact oriented circuits, which implies that, up to normalization, there are two unique
  elements $u_{1},u_{2}\in\ker A$ such that $\sgn(u_{i})=\sigma_{i}$, $i=1,2$.  They satisfy $\Dbarr(u_{i}) = 0$, so
  they are surely not global maximizers.

  There are three sign vectors whose support has cardinality twelve.  Let $\sigma$ be one of these.  Then the
  restriction $\supp(u)\subseteq\supp(\sigma)$ selects a two-dimensional subspace of $\ker A$, and it is easy to see
  that $\Dbarr = 0$ on this subspace.

  There remain 15 sign vectors that have a full support.  For every such sign vector~$\sigma$ the system of the
  algebraic equations in $I_{1}^{\sigma}$ and $I_{2}^{\sigma}$ has to be solved.  To reduce the number of equations and
  the number of variables one may parametrize the solution set $\ker_{\Cb}A$ of $I_{1}^{\sigma}$ by finding a basis
  $u_{1},\dots,u_{5}$ of $\ker A$.  Then this parametrization is plugged into the equations of $I_{2}^{\sigma}$.  Some
  of these systems are at the limit of what today's desktop computer can handle.  Therefore care has to be taken how to
  formulate these equations.  The general strategy is the following:
  \label{page:algorithm}
  \begin{enumerate}
  \item At first, compute a basis $v_{1},\dots,v_{k-1}$ of $K_{\Zb}^{\sigma}$ by using a Gram-Schmidt-like algorithm:
    Renumber the $u_{i}$ such that $d_{\sigma}(u_{5})\neq 0$ and let
    \begin{equation}
      v_{i}:= \frac{d_{\sigma}(u_{5})}{g} u_{i} - \frac{d_{\sigma}(u_{i})}{g} u_{5},
    \end{equation}
    where $g = \gcd(d_{\sigma}(u_{5}),d_{\sigma}(u_{i}))$.
  \item Let $I$ be the ideal in the variables $\lambda_{1},\dots,\lambda_{5}$ generated by the equations
    \begin{equation}
      \prod_{x:v_{i}>0} u(x)^{v_{i}(x)} - \prod_{x:v_{i}<0} u(x)^{-v_{i}(x)}, \text{ for all }i=1,\dots,4,
    \end{equation}
    where $u(x) = \sum_{i=1}^{5}\lambda_{i}u_{i}(x)$.
  \item Compute the saturation $J = (I:\prod_{x\in\Xcal}u(x)^{\infty})$.
  \item Compute the primary decomposition of $J$.
  \end{enumerate}
  Note that the ideal $I$ in the second step corresponds to the ideal $I_{2}(B)$ defined above for the basis
  $B=\{v_{1},\dots,v_{4}\}$, where the variables $u(x)$ have been restricted to the linear subspace $\ker_{\Cb} A$.
  The ideal $J$ obtained by saturation in the third step is then independent of $B$.

  Unfortunately, this simple algorithm does not work for all sign vectors.  Some further tricks are needed to compute
  the primary decomposition within a reasonable time.

  A basis of $\ker A$ is given by the rows $u_{1},\dots,u_{5}$ of the matrix
  \begin{equation*}
    \left(
      \begin{matrix}
        \;1 \phantom{+} {-1} \phantom{+} {-1} \phantom{+} \phantom{-}1 \phantom{+} {-1} \phantom{+} \phantom{-}1 \phantom{+} \phantom{-}1 \phantom{+} {-1} \phantom{+} {-1} \phantom{+} \phantom{-}1 \phantom{+} \phantom{-}1 \phantom{+} {-1} \phantom{+} \phantom{-}1 \phantom{+} {-1} \phantom{+} {-1} \phantom{+} \phantom{-}1 \\
        \;1 \phantom{+} \phantom{-}0 \phantom{+} {-1} \phantom{+} \phantom{-}0 \phantom{+} {-1} \phantom{+} \phantom{-}0 \phantom{+} \phantom{-}1 \phantom{+} \phantom{-}0 \phantom{+} {-1} \phantom{+} \phantom{-}0 \phantom{+} \phantom{-}1 \phantom{+} \phantom{-}0 \phantom{+} \phantom{-}1 \phantom{+} \phantom{-}0 \phantom{+} {-1} \phantom{+} \phantom{-}0 \\
        \;1 \phantom{+} {-1} \phantom{+} \phantom{-}0 \phantom{+} \phantom{-}0 \phantom{+} {-1} \phantom{+} \phantom{-}1 \phantom{+} \phantom{-}0 \phantom{+} \phantom{-}0 \phantom{+} {-1} \phantom{+} \phantom{-}1 \phantom{+} \phantom{-}0 \phantom{+} \phantom{-}0 \phantom{+} \phantom{-}1 \phantom{+} {-1} \phantom{+} \phantom{-}0 \phantom{+} \phantom{-}0 \\
        \;1 \phantom{+} {-1} \phantom{+} {-1} \phantom{+} \phantom{-}1 \phantom{+} \phantom{-}0 \phantom{+} \phantom{-}0 \phantom{+} \phantom{-}0 \phantom{+} \phantom{-}0 \phantom{+} {-1} \phantom{+} \phantom{-}1 \phantom{+} \phantom{-}1 \phantom{+} {-1} \phantom{+} \phantom{-}0 \phantom{+} \phantom{-}0 \phantom{+} \phantom{-}0 \phantom{+} \phantom{-}0 \\
        \;1 \phantom{+} {-1} \phantom{+} {-1} \phantom{+} \phantom{-}1 \phantom{+} {-1} \phantom{+} \phantom{-}1 \phantom{+} \phantom{-}1 \phantom{+} {-1} \phantom{+} \phantom{-}0 \phantom{+} \phantom{-}0 \phantom{+} \phantom{-}0 \phantom{+} \phantom{-}0 \phantom{+} \phantom{-}0 \phantom{+} \phantom{-}0 \phantom{+} \phantom{-}0 \phantom{+} \phantom{-}0
      \end{matrix}
    \right).
  \end{equation*}
  This basis has the following property: Let $u=\sum_{i=1}^{5}\lambda_{i}u_{i}$.  If $\lambda_{j}=0$ for some
  $j=2,3,4,5$, then $\ol D(u) = 0$.  The reason is that if one $\lambda_{j}$ vanishes, then it is easy to see that there
  is bijection between the positive and negative entries of $u$ such that corresponding entries have the same absolute
  value.  This implies that, in order to determine the global maximizer of this model one may saturate $J$ by the
  product $\lambda_{2}\lambda_{3}\lambda_{4}\lambda_{5}$.

  Replacing $J$ by $\left(J:(\lambda_{2}\lambda_{3}\lambda_{4}\lambda_{5})^{\infty}\right)$ makes it possible to solve
  all but one system of equations.  For the last sign vector $\sigma$ a special measure is necessary: The complexity of
  the above algorithm depends on the chosen basis $v_{1},v_{2},v_{3},v_{4}$ of $K^{\sigma}_{\Zb}$.  The $\ell_{1}$-norm
  of each vector $v_{i}$ equals twice the degree of the corresponding equation.  Thus it is advisable to choose the
  vectors $v_{1},v_{2},v_{3},v_{4}$ as short as possible.  As a first approximation, one may try to use a basis of
  \emph{circuit vectors}, i.e., vectors whose support is minimal.  This approach provides a basis
  $v_{1},v_{2},v_{3},v_{4}$ for $K_{\Zb}^{\sigma}$ of the last sign vector, such that the rest of the algorithm sketched
  above works.

  The calculations were performed with the help of Singular\cite{Singular}.  The primary decompositions were done using
  the algorithm of Gianni, Trager and Zacharias (GTZ) implemented in the library \verb|solve.lib|.  Analyzing the
  results yields the following theorem, confirming a conjecture by Thomas Kahle (personal communication):
\end{ex}
\begin{thm}
  \label{thm:4-2bin}
  The binary 4-2 model has, up to symmetry, a unique maximizer of the \KL, which is the uniform distribution over the
  states $0001$, $0010$, $0100$, $1000$ and $1111$.  The maximal value of $\ol D$ is %
  $%\frac{1}{15}(15 \log 3 - 5 \log 5) =
  \log 3 - \frac{1}{3}\log 5 \approx 0.56213298$, it is reached at
  \begin{equation}
    u = \frac{1}{15} (-5, 3, 3, -1, 3, -1, -1, -1, 3, -1, -1, -1, -1, -1, -1, 3).
  \end{equation}
  The maximum value of the $\DE$ is %
  $%\log(1 + 3\cdot 5^{\frac{1}{3}})
  = \log(1 + 3\cdot 5^{\frac{1}{3}})\approx 1.0132035$.
\end{thm}

\section{Computing the projection points}
\label{sec:projpts}

The theory of this paper motivates a second method for computing the maximizers of $\DE$, which is more elementary than
solving the critical equations.  However, knowing the critical equations sheds new light on this method.

Let $P_{+}$ be a projection point and construct $P_{-}$ as in section~\ref{sec:projpoints}.  Then $u = P_{+}-P_{-}$ and
the common $rI$-projection $P_{\Ecal}$ of $P_{+}$ and $P_{-}$ satisfy
\begin{equation}
  \label{eq:projprop}
  u(x) =
  \begin{cases}
    \frac{1}{\mu} P_{\Ecal}(x) & \text{ if }x\in\Zcal,
    \\
    -\frac{1}{1-\mu} P_{\Ecal}(x) & \text{ if }x\notin\Zcal.
  \end{cases}
\end{equation}
On the other hand, $P_{\Ecal}$ lies in the closure of the exponential family.  Suppose that $P_{\Ecal}$ has full
support.  Then the exponential parameterization~\eqref{eq:expfam} implies that there exist
$\alpha_{1},\dots,\alpha_{h}>0$
\begin{equation}
  \label{eq:monomparam}
  P_{\Ecal}(x) = \frac{r(x)}{Z_{\alpha}} \prod_{i=1}^{h} \alpha_{i}^{A_{i,x}}.
\end{equation}

Assume that $\sigma=\sgn(P_{+}-P_{-})$ has full support and define a $(h+1)\times\Xcal$-matrix $A^{\sigma}$ as follows:
Take the matrix $A$ and add a zeroth row with entries
\begin{equation}
  A^{\sigma}_{0,x} := 1 - \sigma_{x} \in \{0,1\}.
\end{equation}
Then equations~\eqref{eq:projprop} and~\eqref{eq:monomparam} together show that $u$ has the form
\begin{equation}
  \label{eq:monomparamsigma}
  u(x) = r(x)\prod_{i=0}^{h}\alpha_{i}^{A^{\sigma}_{i,x}}
\end{equation}
for suitably chosen $\alpha_{i}$.  Here, $\alpha_{0} = -\frac{\mu}{1-\mu} < 0$, and all the other parameters are
positive.  The normalization can be achieved since the row span of $A$ contains the constant vector.  Thus the
projection points, which project into $\Ecal$, can be found by plugging the parameterization~\eqref{eq:monomparamsigma}
into the equation $Au = 0$ and solving for the $\alpha_{i}$.

Again, this method simplifies if $A$ has only integer entries.  Additionaly it is convenient to suppose that $A$ has
only nonnegative entries.  This nonnegativity requirement can always be supposed, since $A$ contains the constant row in
its row span.  In this case the parameterization~\eqref{eq:monomparamsigma} is monomial, so the equation $Au = 0$ is
equivalent to $h$ polynomial equations in the $h+1$ parameters $\alpha_{0},\dots,\alpha_{h}$.

This method is linked to the ideal $I_{2}^{\sigma}$ of the previous section.  As stated there, $I_{2}^{\sigma}$ is
related to the lattice ideal $I_{2,r}^{\sigma}$, which defines a toric variety.  Every toric variety has a monomial
``parameterization'', which induces the monomial parameterization~\eqref{eq:monomparamsigma}.  Unfortunately, in the
general case this monomial parameterization is not surjective.  However, equation~\eqref{eq:monomparamsigma} shows that
it is ``surjective enough'', at least in the case where $\sigma$ has full support.

It is possible to extend this analysis to the case where $\sigma$ does not have full support.  Let $\Ycal =
\supp(\sigma)$.  First it is necessary to parameterize the set $\Ecal^{\Ycal}$ of those probability distributions of
$\ol\Ecal$ whose support is $\Ycal$.
% ***
% One possibility is to replace the matrix $A$ by another matrix such that the
% corresponding monomial parameterization is surjective.  This is always possible (see ***), but in general this leads to
% an increase in the number of parameters.  A better
% ***
One solution is to find an element $r_{\Ycal}\in\ol\Ecal$ such that $\supp(r_{\Ycal}) = \Ycal$.  Then $\Ecal^{\Ycal}$
equals the exponential family over the set $\Ycal$ with reference measure $r_{\Ycal}$ whose sufficient statistics matrix
$A_{\Ycal}$ consists of those columns of $A$ corresponding to $\Ycal\subseteq\Xcal$.  This gives a monomial
parameterization of $\Ecal^{\Ycal}$ with at most $h$ parameters.

The equations obtained from $Au=0$ by plugging in a monomial parameterization for $u$ can be solved by primary
decomposition.  Every solution $(\alpha_{0},\dots,\alpha_{h})$ yields a point of~$X^{\sigma}$.
Theorem~\ref{thm:Dsigmaind} applies in this context.

\begin{ex}
  The above ideas can be used to find the maximizers of the independence model of three random variables of
  cardinalities 2, 3 and 3.  This example is particularly interesting, since the global maximizers are known for those
  independence models where the cardinality of the state spaces of the random variables satisfy an
  inequality\cite{AyKnauf06:Maximizing_Multiinformation}.  The cardinalities 2, 3 and 3 are the smallest set of
  cardinalities that violate this inequality.

  A sufficient statistics of the model is given by
  \begin{equation}
    A_{2-3-3} =
    \left(
    \begin{tabular}{*{18}{c}}
%      & 0 & 1 & 2 & 3 & 4 & 5 & 6 & 7 & 8 & 9 &10 &11 &12 &13 &14 &15 &16 &17 \\
%      \hline
      1 & 1 & 1 & 1 & 1 & 1 & 1 & 1 & 1 & 1 & 1 & 1 & 1 & 1 & 1 & 1 & 1 & 1 \\
      1 & 0 & 0 & 1 & 0 & 0 & 1 & 0 & 0 & 1 & 0 & 0 & 1 & 0 & 0 & 1 & 0 & 0 \\
      0 & 1 & 0 & 0 & 1 & 0 & 0 & 1 & 0 & 0 & 1 & 0 & 0 & 1 & 0 & 0 & 1 & 0 \\
      1 & 1 & 1 & 0 & 0 & 0 & 0 & 0 & 0 & 1 & 1 & 1 & 0 & 0 & 0 & 0 & 0 & 0 \\
      0 & 0 & 0 & 1 & 1 & 1 & 0 & 0 & 0 & 0 & 0 & 0 & 1 & 1 & 1 & 0 & 0 & 0 \\
      1 & 1 & 1 & 1 & 1 & 1 & 1 & 1 & 1 & 0 & 0 & 0 & 0 & 0 & 0 & 0 & 0 & 0
    \end{tabular}
    \right).
  \end{equation}
  The states are numbered in the ternary representation from $000$ to $122$, where the ``highest'' random variable only
  takes two values.  The dimension of the model is $d = 5$ and the state space has cardinality 18.  Thus $\dim\ker A =
  18 - 5 - 1 = 12$.  The symmetry group of the model is generated by the permutation of the two random variables of
  cardinality three and by the permutations within the state space of each random variable.

  The cocircuits can be computed by TOPCOM.  Testing all $3^{18}$ possible sign vectors of length 18 shows that there
  are $182\,796$ non-zero sign vectors in $\ker A$ (up to symmetry).  Checking the support condition~{\ref{it:Var0}.}
  leaves 975 sign vectors.  Excluding all sign vectors where the support of both the negative and the positive part
  exceeds 6 (cf.~Theorem~\ref{thmA:Matus}) reduces the problem to 240 sign vectors.

  The 72 sign vectors that do not have full support can be treated as in the previous section.  For the 168 sign
  vectors that have full support the corresponding systems of equations consist of $\dim\ker A - 1 = 11$ equations of
  $\dim\ker A = 12$ variables.  These are too difficult to solve in this way, but they can be treated using the method
  proposed in this section, which ``only'' requires the primary decomposition of a system of $d=5$ polynomials in
  $d+1=6$ variables.
 
  % Among the remaining sign vectors there is one circuit with support of size 6 which satisfies $\ol D = 0$.  There are
  % three sign vectors whose support has size 9.  One of them yields a solution with $\ol D = \log 2$.  For 45 of the 68
  % sign vectors of support size 12 the ideal $J$ (as defined in the algorithm on page~\pageref{page:algorithm}) is
  % trivial, i.e., there are no solutions.  For 17 further sign vectors there are isolated solutions which have no real
  % points with the correct signs.  Finally there are 6 sign vectors which yield solutions, five of them with $\ol D = 0$
  % and one with $\ol D = \log 2$.

  % The majority of the remaining 168 full sign vectors can be excluded by simple facts such as:
  % \begin{itemize}
  % \item Let $\sgn(u)=\sigma$.  If $\sigma_{x}=\sigma_{y}\neq 0$, then $u(x) + u(y) \neq 0$.
  % \end{itemize}

  The analysis was carried out with the help of Singular.  It proved to be advantageous to use the algorithm of
  Shimoyama and Yokoyama (SY) from the library \verb|solve.lib|.  The following result was obtained:
\end{ex}
\begin{thm}
  The maximal value of $\DE$ for the independence model of cardinalities 2, 3 and 3 equals $\log(3+2\sqrt{2})\approx
  1.7627472$, and the maximal value of $\Dbarr$ is $\log(2(1+\sqrt{2}))\approx 1.5745208$.  Up to symmetry there is a
  unique global maximizing probability distribution
  \begin{equation}
    (1 - \frac{\sqrt2}2)(\delta_{012} + \delta_{020}) + (\sqrt2 - 1)\delta_{100}.
  \end{equation}
\end{thm}

In order to compare the two methods of finding the maximiers of $\Dbarr$ resp.~$\DE$ presented in this section and in
the last section let $d$ be the dimension of the model and let $r=\dim\ker A$.  All algorithms are most efficient if $A$
is chosen such that $h = d+1$.  Then, for any sign vector $\sigma$ with full support, the algorithm on
page~\ref{page:algorithm} starts with $r-1$ equations (corresponding to a basis of $K_{\Cb}^{\sigma}$) in $r$ variables
$\lambda_{1},\dots,\lambda_{r}$, which are then saturated.  On the other hand, the method in this section starts with
the $d+1$ equations $A u = 0$ in the $d+2$ variables $\alpha_{0},\dots,\alpha_{d+1}$.  Thus, generically, the first
method should perform better when the codimension of the model is small, while the second method should perform better
when the dimension of the model is small.

\section{Conclusions}
\label{sec:conclusions}

In this work a new method for computing the maximizers of the \KL{} from an exponential family $\Ecal$ has been
presented.  The original problem of maximizing $\DE$ over the set of all probability distributions is transformed into
the maximization of a function $\Dbarr$ over $\ker A$, where $A$ is the sufficient statistics of $\Ecal$.  It has been
shown that the global maximizers of both problems are equivalent.
%Furthermore, the \emph{quasi-critical points} of both problems are equivalent, where the quasi-critical points are defined as the solutions to all equations which arise from the first order conditions of the corresponding problem.
Furthermore, every local maximizer of $\DE$ yields a maximizer of $\Dbarr$.  At present it is not known whether the
converse statement also holds.

The two main advantages of the reformulation are:
\begin{enumerate}
\item A reduction of the dimension of the problem.
\item The function $\Dbarr$ can be computed by a formula.
\end{enumerate}
If $\Ecal$ has codimension one, then the first advantage is most visible.  Even this simple case can be useful in order
to obtain examples of maximizers having specific properties.

The maximizers of $\Dbarr$ can be computed by solving the critical equations.  These equations are nice if they are
considered separately for every sign vector $\sigma$ occuring in $\ker A$.  There are some conditions which allow to
exclude certain sign vectors from the beginning.  If the matrix $A$ contains only integer entries, then the critical
equations are algebraic, once the sign vector is fixed.  In this case tools from commutative algebra can be used to
solve these equations.
%  Theorem\ref{thm:Dsigmaind} shows that 

A second possibility is to compute the points satisfying the projection property.  If $A$ is an integer matrix and if
the sign vector is fixed, then one obtains algebraic equations which are related to the critical equations of $\ol D$.
This method is more appropriate for exponential families of small dimension.

Of course, a problem with these two approaches is that every sign vector needs to be treated separately, and their
number grows quickly.  By contrast, the problem of finding the maximizers of $\DE$ becomes a smooth problem if one
restricts the support of the possible maximizers.  In general the set of possible support sets is much smaller than the
set of sign vectors.
%
% (Note that we only need to consider support vectors which appear in some sign vector which has not been excluded by any
% criterion.  Now every sign vector produces \emph{two} support sets.  It may happen that only one of these satisfies the
% support constraint of Theorem~\ref{thmA:Matus}.  On the other hand different sign vectors may produce the same support.
% Finally the symmetry group of the problem may further reduce the number of support sets which have to be considered.).
Still, two examples have been given where the maximizers where not known before and where the separate analysis of each
sign vector was feasible.

% Of course, these two advantages covers only some aspects of the two maximization problems.  As a disadvantage
%This has been demonstrated by an example of a two-dimensional exponential family on a state space of cardinality 5.
%The theoretical results of this work mean that solving the critical equations of $\Dbar$ yields just as many solutions
%as solving the critical equations of $\DE$.  However, it might be that there exist more local maxima of $\Dbar$, which
%would be a disadvantage when using gradient-search algorithms.

% Some open questions concerning the correspondence between the two maximization problems remain.  As mentioned above, it
% is not known whether there is a bijection between the local maximizers of both problems.  Furthermore, the problem of
% maximizing $\Dbar$ made use of the projection property.  However, 

%---------------------------------------------------------------------------------------

\appendix

\section*{Acknowledgement}

I would like to thank František Matúš for his interest and many helpful remarks.  I also thank Thomas Kahle, Bastian
Steudel and Nihat Ay for numerous inspiring discussions.

\bibliographystyle{plain}
\bibliography{general}

\end{document}